\documentclass[twocolumn]{autart}    % Enable this line and disable the 
\usepackage{graphicx}          % Include this line if your 

\usepackage{cite}
\usepackage{amsmath,amssymb,amsfonts}
\usepackage{graphicx}
\usepackage{textcomp}

\usepackage{graphics}
\usepackage{epstopdf}
\usepackage{epsfig}
\usepackage{mathrsfs}
\usepackage{color}
\usepackage{array}
\usepackage{booktabs}
\usepackage{lipsum,multicol}
\usepackage{subfigure}
\usepackage{times}

\usepackage{algorithm}
\usepackage{pdfsync}
\usepackage{pgfplots}
\usepackage{pgfplotstable}
\usepackage[subnum]{cases}
\usepackage{multirow}
\usepackage{algpseudocode}
\usepackage{float}
\usepackage{setspace}
\usepackage{bm}
\usepackage{multicol}
\usepackage{epstopdf}
\usepackage[center]{caption}
\newcommand{\beq}{\begin{equation}}
	\newcommand{\eeq}{\end{equation}}

\newcommand{\abs}[1]{\left|#1\right|} % absolute value
\renewcommand{\vec}[1]{\ensuremath{\boldsymbol{#1}}}

\newenvironment{proof}{\par\noindent\textbf{Proof.} }{\hfill $\square$\par}
\def\BibTeX{{\rm B\kern-.05em{\sc i\kern-.025em b}\kern-.08em
		T\kern-.1667em\lower.7ex\hbox{E}\kern-.125emX}}
\algdef{SE}[REPEATUNTIL]{RepeatUntil}{EndRepeat}[1]{\qquad \; \algorithmicrepeat\ #1}{\qquad \; \algorithmicuntil\ \text{Convergence}}
\begin{document}

\begin{frontmatter}

\title{\textcolor{blue}{A Volumetric Privacy Measure for Dynamical Systems With Bounded Disturbance}} % Title, preferably not more 

\thanks{The work was partially supported by the Research Grants Council of Hong Kong under Project CityU 21208921, a grant from Chow Sang Sang Group Research Fund sponsored by Chow Sang Sang Holdings International Limited.}

\author[CityU]{Chuanghong Weng}\ead{cweng7-c@my.cityu.edu.hk},    % Add the 
\author[CityU]{Ehsan Nekouei}\ead{enekouei@cityu.edu.hk},               % e-mail address 

\address[CityU]{Department of Electrical Engineering, City University of Hong Kong, Hong Kong, China}  % Please

\begin{keyword}                           % Five to ten keywords,  
Volumetric privacy measure; privacy protection; interval analysis; bounded disturbance.               % chosen from the IFAC 
\end{keyword}                             % keyword list or with the 
                                          % help of the Automatica 
                                          % keyword wizard

\begin{abstract}                          % Abstract of not more than 200 words.
\textcolor{blue}{In this paper, we first present a volumetric privacy measure for dynamical systems with bounded disturbances, wherein the states of the system contain private information and an adversary with access to sensor measurements attempts to infer the set of potential values of the private information. Under the proposed privacy measure, the volume of the uncertainty set of the adversary given the sensor measurements is considered as the privacy level of the system. We next characteristic the time evolution of the proposed privacy measure and study its properties for a particular system with both public and private states, where a set containing the public state is shared as the observation. Approximate set-membership estimation techniques are developed to compute the private-state uncertainty set, and the properties of the privacy measure are analyzed, demonstrating that the uncertainty reduction of the adversary is bounded by the information gain from the observation set. Furthermore, an optimization-based privacy filter design problem is formulated, employing randomization and linear programming to enhance the privacy level. The effectiveness of the proposed approach is validated through a production–inventory case study. Results show that the optimal privacy filter significantly improves robustness against inference attacks and outperforms two baseline mechanisms based on additive noise and quantization.}
\end{abstract} %An adversary is assumed to execute an inference attack by exploiting the observed public state set to estimate an uncertainty set for the private state. The volume of this inferred set quantifies the adversary’s estimation uncertainty and serves as the proposed volumetric privacy metric.
\end{frontmatter}

\section{Introduction}
\subsection{Motivation}
\textcolor{blue}{Data sharing plays a pivotal role in enabling cooperative decision-making and optimization in dynamic processes. However, the exposure of such data may inadvertently reveal sensitive information\cite{li2024preserving}. Dynamical systems subject to bounded disturbances without knowledge of their underlying distributions provide a natural framework for modeling numerous practical applications involving sensitive information. Despite their importance, the notion of privacy in such systems remains insufficiently explored.} %Specifically, correlations between shared metrics and underlying operational information can be exploited by adversaries to develop competitive and malicious strageties. This challenge highlights the critical need for methodologies that preserve data utility while ensuring privacy protection for dynamic systems.
%\textcolor{blue}{}

%\textcolor{blue}{Motivated by these considerations, this paper investigates the notion of volumetric privacy for systems affected by bounded disturbance. We develop privacy-preserving strategies aimed at maximizing an adversary's uncertainty, i.e., the volume of uncertainty set, in inferring private states. The proposed approaches are applicable to both deterministic and stochastic systems with bounded disturbance, without requiring prior knowledge of the underlying probability distributions.}
\textcolor{blue}{Within the context of set-membership estimation, the states of systems with bounded disturbance can be represented by geometric sets, such as ellipsoids or zonotopes, whose volumes quantify the degree of inference uncertainty. Motivated by these considerations, this paper investigates the notion of volumetric privacy for systems affected by bounded disturbance. We develop privacy-preserving strategies aimed at maximizing an adversary's uncertainty, i.e., the volume of uncertainty set, in inferring private states. The proposed approaches are applicable to both deterministic and stochastic systems, without requiring prior knowledge of the underlying probability distributions.}
\subsection{Related Work}
\textcolor{blue}{Stochastic approaches to privacy primarily include differential privacy and information-theoretic methods.
 Differential privacy (DP) \cite{dwork2014algorithmic} has been incorporated into dynamic settings through differentially private Kalman filtering \cite{le2013differentially}, DP-preserving average consensus via noise injection \cite{mo2016privacy}, and minimal-noise mechanisms for multi-agent systems based on observability properties \cite{zhang2022much}. Recent work \cite{li2024preserving} introduced a trace-based variance–expectation ratio to quantify topology preservation and derived optimal noise designs, while \cite{hassan2019differential} provided a comprehensive survey of DP in dynamical systems. In parallel, information-theoretic approaches quantify privacy leakage using mutual information or conditional entropy. Recent studies include mutual-information-based private filtering for hidden Markov models \cite{cavarec2021designing}, directed-information-based privacy filters for linear systems \cite{tanaka2017directed,nekouei2019information}, and recent extensions to partially observable Markov decision processes addressing privacy-aware estimation and control \cite{molloy2023smoother,weng2025optimal}.}

\textcolor{blue}{Most existing studies focus on deterministic or stochastic systems with unbounded noise and known distributions, leaving privacy protection for systems subject to unknown-but-bounded disturbances relatively unexplored. Recent works \cite{dawoud2023differentially,dawoud2024privacy} developed differentially private set-based estimators using truncated noise, while \cite{khajenejad2023guaranteed} introduced guaranteed privacy concepts and optimization methods for $\mathcal{H}_\infty$-based privacy-preserving interval observers. State-opacity-based methods \cite{saboori2007notions,liu2020verification} ensured indistinguishable outputs between secret and non-secret states but did not quantify the associated estimation uncertainty. Note that set-membership estimators typically characterize uncertainty through bounded geometric sets such as intervals~\cite{jaulin2001interval}, zonotopes~\cite{le2013zonotopes}, or ellipsoids~\cite{chern2005ellips}, where the corresponding set volume naturally describes the amount of estimation uncertainty. Motivated by this, we analyze privacy leakage in systems with private and public states and propose a volumetric approach that maximizes the private-state set volume, thereby enhancing privacy while explicitly accounting for geometric effects under inference attacks.}

\textcolor{blue}{There are some related deterministic approaches to privacy without adding noise, \emph{e.g.}, plausible deniability \cite{monshizadeh2019plausible} and noiseless privacy \cite{farokhi2021noiseless}. In \cite{monshizadeh2019plausible}, privacy leakage in deterministic systems was measured by the volume of reachable state sets, and was determined by the observability. However, this framework does not apply to systems with bounded disturbance, where inference uncertainty depends on both observability and disturbance. Moreover, the problem of privacy filter design was not addressed in \cite{monshizadeh2019plausible}, whereas we propose a concrete design using randomization and optimization. The work in \cite{silvestre2023privacy} addressed parameter privacy in deterministic systems via constrained convex generators (CCGs), which differs from our private state protection setting. Also, defense strategies in \cite{silvestre2023privacy} involve ceasing information sharing or altering parameters, which might be unsuitable for fixed-parameter systems with continuous communication. As discussed in Sec. \ref{Sec.InfCCG}, the complexity of CCG-based inference grows exponentially with time, motivating our use of interval analysis for computational efficiency.}

%\textcolor{blue}{There are some related deterministic approaches to privacy without adding noise, \emph{e.g.}, plausible deniability \cite{monshizadeh2019plausible} and noiseless privacy \cite{farokhi2021noiseless}. In \cite{monshizadeh2019plausible}, privacy leakage in deterministic systems was measured by the volume of reachable state sets, and was determined by the observability. However, this framework does not apply to systems with bounded disturbance, where inference uncertainty depends on both observability and disturbance. Moreover, the problem of privacy filter design was not addressed in \cite{monshizadeh2019plausible}, whereas we propose a concrete design using randomization and optimization. The work in \cite{silvestre2023privacy} addressed parameter privacy in deterministic systems via constrained convex generators (CCGs), which differs from our private state protection setting. Also, defense strategies in \cite{silvestre2023privacy} involve ceasing information sharing or altering parameters, which might be unsuitable for fixed-parameter systems with continuous communication. As discussed in Sec. \ref{Sec.InfCCG}, the complexity of CCG-based inference grows exponentially with time, motivating our use of interval analysis for computational efficiency.}

\textcolor{blue}{Noiseless privacy \cite{farokhi2021noiseless} and non-stochastic privacy \cite{farokhi2019development} employed non-stochastic information-theoretic approaches to limit information leakage, assuming static private-variable domains and without accounting for temporal dependencies in sequential data. While noiseless privacy, non-stochastic privacy, and our volumetric privacy all achieve privacy through the release of bounded outputs, in our setup, dynamical systems subject to bounded disturbance have time-variant private-state reachable sets that can be recursively estimated, enabling dynamic leakage evaluation. Building on this insight, the proposed volumetric privacy filter dynamically evaluates the private state set and adapts the observation accordingly, thereby achieving higher privacy levels with lower data distortion, as shown in Sec. \ref{Sec.NumericalVer}.}%As shown in Sec. \ref{Sec.NumericalVer}, our volumetric privacy framework offers greater flexibility for privacy–utility trade-offs than other methods.}

\textcolor{blue}{Finally, other deterministic privacy-preserving approaches often exploit observability reduction or state decomposition to protect private information in multi-agent systems. For example, the authors in \cite{pequito2014design} established a connection between network privacy and its observability space, proposing a privacy-aware communication protocol that achieves average consensus while protecting initial states. The authors in \cite{ramos2024privacy,ramos2023trade} presented privacy-preserving consensus algorithms that incorporate augmented states and novel consensus mechanisms, balancing accuracy, resilience, and privacy	guarantees simultaneously. State decomposition methods, as in \cite{wang2019privacy}, split each node’s state into randomized components to prevent disclosure of individual states during consensus.}
\subsection{Contributions}
%This paper addresses the privacy protection problem for dynamic systems with UBB noise, as illustrated in Fig. \ref{Fig.SystemSetup}. In this framework, the system state is partitioned into two categories: the public state $X_k$ and the private state $Y_k$, both of which belong to a known bounded set. The observation of the system states is the set $\mathcal{M}_{k|k}^x$ which contains the actual public state $X_k=x_k$. An untrusted third party, i.e., the adversary, uses $\mathcal{M}_{k|k}^x$ to infer the private state by constructing an uncertainty set $\mathcal{Y}_{k|k}$, which is referred to as the inference attack. 
\textcolor{blue}{In this paper, we investigate volumetric privacy in dynamical systems subject to bounded disturbances, as illustrated in Fig.~\ref{Fig.SystemSetup}. The system state $S_k$ contains both utility information $X_k$ and private information $Y_k$, while an adversary exploits the observation set $\mathcal{M}_{k|k}^x$ to infer the private information $Y_k$ via the associated uncertainty set $\mathcal{Y}_{k|k}$.}
\begin{figure}[h]
	\centering
	\includegraphics[width=3.0in]{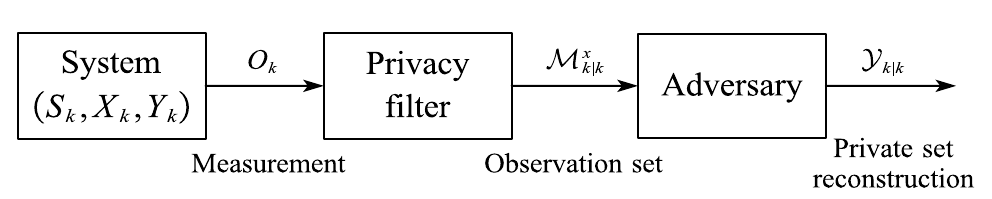}
	\caption{The system setup. }
	\label{Fig.SystemSetup}
\end{figure}

\textcolor{blue}{The primary contribution of this work is the development of an extensible framework for privacy analysis and mitigation in dynamic systems subject to bounded disturbance. This contribution can be summarized in three principal aspects. (1) Volumetric Privacy Measure: We introduce a privacy metric based on the volume of the estimated private-state set obtained via set-membership estimation given the available observations. (2) Privacy Level Computation: We develop computational methods to quantify privacy level and prove that the relevant privacy leakage is bounded by the information gain from the observations. (3) Optimal Privacy Filter: Since the inappropriate choice of the observation set would lead to privacy leakage of the private state, we design a randomized, optimization‐based filter that perturbs and then refines observations to maximize inference uncertainty of attackers. Finally, the proposed framework is demonstrated on a production–inventory case study, showing that our privacy filter significantly reduces the adversary’s capability to estimate the private production rate.}

%The contributions of this paper are summarized as follows:
%\begin{itemize}
%	\item Volumetric Privacy Measure: A privacy measure is proposed based on the volume of the uncertainty set for the private state.
%	\item Privacy Level Computation: We provide computational methods utilizing constrained convex generator and interval analysis to evaluate the privacy level. We studied the properties of interval inference and show the volumetric privacy level is bounded by the new information extracted from the observation set. Also, a smaller privacy level avoid the attack to update its central estimation of the private state.
%	\item Optimal Privacy Filter Design: An optimization-based approach is introduced to design privacy filters, where we first randomize the observation and then optimize it to enhance privacy. We prove that the optimization step cannot be reversed to obtain a smaller inference set due to the randomization operation.
%	\item Simulation Demonstration: We validate our approach using the production-inventory example. The results demonstrate that the adversary can infer sensitive production rates through inventory data, and the leakage can be significantly mitigated under the proposed optimal privacy filter design.  
%\end{itemize}
\subsection{Outline}
\textcolor{blue}{The remainder of the paper is organized as follows. Section~\ref{Sec.ModelMeasure} introduces the system model and defines the volumetric privacy and utility measures. Section~\ref{Sec.InferAttack} formalizes the inference attack, presents computational approaches for evaluating the privacy level, and discusses the properties of the proposed measure. Section~\ref{Sec.PolicyDesign} develops an optimal privacy filter to mitigate privacy leakage while preserving a desired utility level. Numerical results are presented in Section~\ref{Sec.NumericalVer}, followed by concluding remarks in Section~\ref{Sec.Conclusion}.}
\subsection{Notation}
We use italic letters to denote the set of unknown variables, \emph{e.g.}, $\mathcal{X}$ and $\mathcal{Y}$ for $X$ and $Y$. For the non-interval set $\mathcal{Z}$, we use $A\mathcal{Z}$ to denote the set $\left\{AZ | Z\in \mathcal{Z}\right\}$, and use $\mathcal{Z} \oplus \mathcal{R}$ to represent $\left\{Z\!+\!R | Z\in \mathcal{Z}, R \in \mathcal{R}\right\}$. Furthermore, the 1-norm of the column vector $b$ with n dimensions is defined as $\left\|b\right\|_1 = \sum_{i=1}^{n} \abs{b\left(i\right)}$ with the absolute value $\abs{b\left(i\right)}$, and $b^\top$ is the transpose of $b$. The 1-norm of the matrix $A$ is defined as $\left\| A \right\|_1 \!\frac{\Delta}{=} \! \sum_{i,j}^{n,m}{\left|a_{i,j}\right|}$. The vector $\mathbf{1}_{n_x}$ denotes a column vector of ones with $n_x$ dimensions, while $I_{n_x\!\times\! n_x}$ represents an identity matrix of size $n_x \times n_x$. \textcolor{blue}{The operator $\mathrm{diag}(v)$ denotes a diagonal matrix constructed from the vector $v$.}

\section{System Model and Privacy Measure}\label{Sec.ModelMeasure}
\subsection{System model}
\textcolor{blue}{Consider the following stable dynamical system
\begin{align}
	\mathbf{G}: \left\{ 
	\begin{array}{l}
		S_k = f\left(S_{k-1}, W_k\right),\\[3pt]
		O_k = g\left(S_k, V_k\right),
	\end{array}
	\right.
\end{align}
where $S_k$ denotes the system state, with initial condition $S_0$ belonging to a bounded set $\mathcal{S}_0 \subset \mathcal{R}^s$, $O_k$ is the sensor measurement. The unknown disturbances $W_k$ and $V_k$ are assumed to be bounded within the sets $\mathcal{W}_k \subset \mathcal{R}^s$ and $\mathcal{V}_k \subset \mathcal{R}^o$, respectively. The system state $S_k$ contains both private and utility-related information, which are represented as continuous variables
\begin{align}
	Y_k = h(S_k), \qquad X_k = u(S_k),
\end{align}
where $Y_k$ represents the private information that must be kept confidential, and $X_k$ corresponds to the utility information intended to be disclosed. We assume that the adversary has
full knowledge of system model $\mathbf{G}$ and will collect observations of the system to infer the private information.}

\textcolor{blue}{A special case of this model is the following linear non-Gaussian system with invertible $A_1$ and $A_2$,
\begin{align}\label{Eq.LinearSys}
	\mathbf{G}_1:\left\{ \begin{array}{c}
		X_k=A_1X_{k-1}+A_2Y_{k-1}+B_1W_{k}^{x}\\
		Y_k=A_3X_{k-1}+A_4Y_{k-1}+B_2W_{k}^{y}\\
		O_k=X_k\\
	\end{array} \right. ,
\end{align}
where the system state $S_k = [X_k^\top,\, Y_k^\top]^\top$ consists of the public state $X_k\in\mathcal{R}^{n_x}$ and the private state $Y_k\in\mathcal{R}^{n_x}$, and the measurement $O_k$ corresponds to the public state $X_k$.  The initial public and private states belong to $\mathcal{X}_{0|-1}$ and $\mathcal{Y}_{0|-1}$, respectively.}
\subsection{Motivating Examples}
\textcolor{blue}{We next consider two motivating examples to illustrate the necessity of protecting privacy of systems with bounded disturbance.}

\textcolor{blue}{\textbf{Production-inventory system}: In supply chain management \cite{ortega2004control,sana2010production}, the inventory level $X_k$ and the production rate $Y_k$ evolves according to $\mathbf{G}_1$. While firms may disclose inventory information $X_k$ to distributors to boost sales, the production rate $Y_k$ contains sensitive strategic information such as production efficiency and supply chain operations. Since $X_k$ and $Y_k$ are correlated, releasing $X_k$ directly risks revealing private production details. Therefore, it is necessary to transform or mask observations to preserve the privacy of $Y_k$ while maintaining the utility of public inventory data $X_k$.}

\textcolor{blue}{\textbf{Traffic management system}: In intelligent transportation, vehicles may report their velocities to a central controller to optimize traffic flow, \emph{e.g.}, by adjusting the speed limit on highways. Given bounded disturbances from environmental factors like uneven ground, the vehicle dynamics fit the model $\mathbf{G}_1$ with unknown-but-bounded disturbance. Here, velocity $X_k$ can be considered public data used for traffic management, while position $Y_k$ is private, as it can be used to identify individual vehicles. To protect location privacy, vehicles may intentionally report blurred or randomized velocity observations that preserve system utility but reduce the risk of precise location inference.}
\subsection{Privacy and Utility Measures}\label{Sec.PrivacyUtilityMeasure}
\textcolor{blue}{Notably, in differential privacy for dynamical systems, the ranges of utility and private information are typically assumed to be time-invariant and unbounded, and privacy is characterized by probabilistic indistinguishability. However, for systems subject to bounded disturbances, the utility and private information can be represented by bounded uncertainty sets that are updated online using set-membership estimation. Accordingly, it is desirable to introduce a privacy measure tailored to such set-based descriptions.}

\textcolor{blue}{Given the sensor measurements, an adversary can employ set-membership estimation techniques, \emph{e.g.}, \cite{le2013zonotopes,chern2005ellips}, to estimate both the utility and private information by constructing uncertainty sets $\mathcal{X}_{k|k}$ and $\mathcal{Y}_{k|k}$ that contain the true values $X_k = x_k$ and $Y_k = y_k$, respectively. We would further provide numerical approaches to illustrate how the adversary estimates $\mathcal{X}_{k|k}$ and $\mathcal{Y}_{k|k}$ in next section. This inference procedure, referred to as an \emph{inference attack}, results in privacy leakage, as the adversary's confidence in the private information increases when the uncertainty set $\mathcal{Y}_{k|k}$ shrinks. In the extreme case where $\mathcal{Y}_{k|k} = \{y_k\}$, the private value is fully revealed. Since $Y_k$ resides in a continuous space, the cardinality of $\mathcal{Y}_{k|k}$ is not meaningful, motivating a geometric approach to quantify uncertainty.}

\textcolor{blue}{We therefore employ the volume of the uncertainty set as a quantitative measure of privacy:
\begin{equation}
	\mathrm{Vol}\!\left(\mathcal{Y}_{k|k}\right) = \int_{\mathcal{Y}_{k|k}} \mathrm{d}y,
\end{equation}
where $\mathrm{Vol}(\mathcal{Y}_{k|k})$ denotes the Lebesgue measure of the set $\mathcal{Y}_{k|k} \subseteq \mathcal{R}^{n_y}$. The privacy measure at time $k$ is then defined as
\begin{equation}
	P_k\!\left(\mathcal{Y}_{k|k}\right) := \mathrm{Vol}\!\left(\mathcal{Y}_{k|k}\right).
	\label{eq:privacy_measure}
\end{equation}
A smaller volume corresponds to reduced privacy, while $\mathrm{Vol}(\mathcal{Y}_{k|k}\!)\!=\!0$ indicates complete exposure.}

\textcolor{blue}{Similarly, the utility information can be represented by the uncertainty set $\mathcal{X}_{k|k}$. Larger $\mathcal{X}_{k|k}$ implies greater uncertainty in recovering the true public information $X_k = x_k$, resulting in higher distortion. Accordingly, we define the utility measure as
\begin{equation}
	U_k\!\left(\mathcal{X}_{k|k}\right) := \frac{1}{\mathrm{Vol}\!\left(\mathcal{X}_{k|k}\right)},
	\label{eq:utility_measure}
\end{equation}
so that higher $U_k$ corresponds to better utility. Throughout, we assume that all relevant uncertainty sets are measurable with finite, positive volume.}

\textcolor{blue}{In the following, we analyze the proposed volumetric privacy measure in the context of the linear non-Gaussian system $\mathbf{G}_1$, and we develop practical approximations for evaluating the resulting privacy leakage. Based on this analysis, we design a privacy filter that mitigates leakage while satisfying utility requirements. A extension study of volumetric privacy for general nonlinear systems is left for future work.}
\section{Inference Attack for Linear Systems} \label{Sec.InferAttack}
\subsection{Inference Attack via Set Operations}
\textcolor{blue}{We assume that the adversary observes a set of public states, denoted by $\mathcal{M}_{k|k}^x$, which is generated by the privacy filter, as shown in Fig. \ref{Fig.SystemSetup}. The observation set contains the true public state $X_k = x_k$ along with additional elements intended to obfuscate the adversary's estimation of the private states. The adversary then performs the inference attack by identifying all possible values of the private state that are consistent with $\mathcal{M}_{k|k}^x$, thereby constructing its uncertainty set. In the following, we define the inference attack recursively.}  

\textcolor{blue}{At time $k$, given the public state set $\mathcal{X}_{k-1|k-1}$ and the uncertainty private state set $\mathcal{Y}_{k-1|k-1}$, the set of states can be predicted based on the system model \eqref{Eq.LinearSys}, i.e.,
\begin{align} 
		\mathcal{X} _{k|k-1}&=A_1\mathcal{X} _{k-1|k-1}\oplus A_2\mathcal{Y} _{k-1|k-1}\oplus B_1\mathcal{W} _{k}^{x}, \label{Eq.calX1} \\
		\mathcal{Y} _{k|k-1}&=A_3\mathcal{X} _{k-1|k-1}\oplus A_4\mathcal{Y} _{k-1|k-1}\oplus B_2\mathcal{W} _{k}^{y}. \label{Eq.calY1}
\end{align}
After receiving the observation set of the public state $\mathcal{M}^x_{k|k} \subseteq \mathcal{X} _{k|k-1}$, the adversary extracts new information from $\mathcal{M}^x_{k|k}$ and updates the uncertainty sets of $X_{k-1}$ and $Y_{k-1}$ via the following steps,
\begin{align} 
	\!\!\!\!\mathcal{M} _{k\!-\!1|k}^{x}\!\!=\!A_{1}^{\!-\!1} \!\mathcal{M} _{k|k}^{x}&\!\oplus\! \left(\! -\!A_{1}^{\!-\!1}\!A_2 \!\right)\! \mathcal{Y} _{k\!-\!1|k\!-\!1}\!\oplus \!\left( -\!A_{1}^{\!-\!1}\!B_1 \!\right) \!\mathcal{W} _{k}^{x} \!,\!\! \label{Eq.exactYEst0} \\
	\!\!\!\!\mathcal{M} _{k\!-\!1|k}^{y}\!=\!A_{2}^{\!-\!1} \! \mathcal{M} _{k|k}^{x}&\!\oplus\! \left(\! -A_{2}^{\!-\!1}\!A_1 \!\right)\! \mathcal{X} _{k\!-\!1|k\!-\!1}\!\oplus \!\left(\! -A_{2}^{\!-\!1}\!B_1 \!\right)\! \mathcal{W} _{k}^{x} \!,\!\! \label{Eq.exactYEst1} \\
	\mathcal{X} _{k-1|k}&=\mathcal{M} _{k-1|k}^{x}\cap \mathcal{X} _{k-1|k-1},  \label{Eq.exactYEst2} \\
	\mathcal{Y} _{k-1|k}&=\mathcal{M} _{k-1|k}^{y}\cap \mathcal{Y} _{k-1|k-1}, \label{Eq.exactYEst3}
\end{align}
where it first computes the possible sets of the public and private states, i.e., $\mathcal{M} _{k-1|k}^{x}$ and $\mathcal{M} _{k-1|k}^{y}$, based on the system model \eqref{Eq.LinearSys} and the observation $\mathcal{M}^x_{k|k}$ in \eqref{Eq.exactYEst0} and \eqref{Eq.exactYEst1}, and then reduces the uncertainty sets of $X_{k-1}$ and $Y_{k-1}$ via intersection operations in \eqref{Eq.exactYEst2} and \eqref{Eq.exactYEst3}.}

According to the system dynamics \eqref{Eq.LinearSys}, the adversary estimates the uncertainty set of $Y_k$ via the following forward inference,
\begin{align} \label{Eq.exactYEst4}
	\mathcal{Y} _{k|k}=A_3\mathcal{X} _{k-1|k}\oplus A_4\mathcal{Y} _{k-1|k}\oplus B_2\mathcal{W} _{k}^{y}.
\end{align}
Finally, the public state set can be further calibrated via,
\begin{align} 
	&\mathcal{X}_{k|k}=\mathcal{M}^x _{k|k}\cap \mathcal{M}^x_{k|k-1}, \label{Eq.exactYEst5} \\
	\mathcal{M}^x_{k|k-1}=&A_1\mathcal{X} _{k-1|k}\oplus A_2\mathcal{Y} _{k-1|k}\oplus B_1\mathcal{W} _{k}^{x}, \label{Eq.exactYEst6}
\end{align}
where $\mathcal{M}^x_{k|k-1}$ is the predicted uncertainty set of $X_k$ based on the calibrated sets $\mathcal{X} _{k-1|k}$ and $\mathcal{Y} _{k-1|k}$.

Starting from $k=0$, with the initial uncertainty sets $\mathcal{X}_{0|-1}$ and $\mathcal{Y}_{0|-1}$, the adversary can recursively update the uncertainty sets of $X_k$ and $Y_k$ via the backward calibration \eqref{Eq.exactYEst0}-\eqref{Eq.exactYEst3}, and the forward inference \eqref{Eq.exactYEst4}-\eqref{Eq.exactYEst6}. The backward calibration \eqref{Eq.exactYEst0}-\eqref{Eq.exactYEst3} reduces the uncertainty of $X_{k-1}$ and $Y_{k-1}$, which leads to the following proposition. 
\begin{prop}\label{Pro.InfPro}
	For any $k \geqslant 1$, the adversary's uncertainty set for the private state \eqref{Eq.exactYEst4} is a subset of its corresponding prediction set \eqref{Eq.calY1}, i.e., $\mathcal{Y}_{k|k} \subseteq \mathcal{Y}_{k|k-1}$. Moreover, given uncertainty sets $\mathcal{X}_{k-1|k-1}$ and $\mathcal{Y}_{k-1|k-1}$ that contain the true system states $X_{k-1} \!=\! x_{k-1}$ and $Y_{k-1} \!=\! y_{k-1}$, if the observation set $\mathcal{M}^x_{k|k}$ contains the true public state $X_{k} \!=\! x_k$, then the inference set $\mathcal{Y}_{k|k}$ contains the true private state $Y_{k} \!= \!y_k$.
\end{prop}
\textcolor{blue}{\begin{proof}
	Since $\mathcal{M}^x_{k|k}$ contains $x_{k}$, it follows from \eqref{Eq.exactYEst0} that $x_{k-1} \!\in\! \mathcal{M}^x_{k-1|k}$. As $x_{k-1}$ also belongs to $\mathcal{X}_{k-1|k-1}$, their intersection $\mathcal{X}_{k-1|k}$ necessarily contains $x_{k-1}$. By the same reasoning, $y_{k-1} \!\in\! \mathcal{Y}_{k-1|k}$. Propagating through the system dynamics $\mathbf{G}_1$ yields $y_{k} \!\in \!\mathcal{Y}_{k|k}$. Finally, since $\mathcal{X}_{k-1|k} \!\subseteq\! \mathcal{X}_{k-1|k-1}$ and $\mathcal{Y}_{k-1|k} \!\subseteq\! \mathcal{Y}_{k-1|k-1}$, we have $\mathcal{Y}_{k|k} \!\subseteq \!\mathcal{Y}_{k|k-1}$.
\end{proof}}
According to Proposition \ref{Pro.InfPro}, the adversary can reduce its uncertainty of the private state via the inference attack since it can obtain a smaller uncertainty private state set $\mathcal{Y}_{k|k}$ that contains the actual private state $y_k$. In particular, if the uncertainty set $\mathcal{Y}_{k|k}$ contains only one element, then the adversary can obtain the actual private state. 
%\section{Inference Attack Approximation}\label{Sec.InfAttackCmp}
%\subsection{Inference Attack Approximation}\label{Sec.InfAttackCmp}

\textcolor{blue}{As addressed in existing set-membership estimation approaches~\cite{le2013zonotopes,chern2005ellips}, the set operations involved in inference attacks can be computationally expensive. To mitigate this complexity, uncertainty sets are often restricted to specific geometric forms, enabling more efficient implementation of set operations. However, more complex representations generally incur higher computational costs in volume evaluation. In the following, we present two approximation methods for implementing inference attacks and analyze their computational complexity, based on which we establish properties of volumetric privacy.}
\subsection{Inference Attack Approximation via CCGs}\label{Sec.InfCCG}
\textcolor{blue}{In this subsection, we show that the inference attack can be approximated using the constrained convex generator (CCG), a general set representation proposed in \cite{silvestre2021constrained}.
\begin{defn}[CCG Representation]\cite{silvestre2021constrained}
	The constrained convex generator $\mathcal{Z} =\left( G,c,A,b,\mathcal{C} \right) \subset \mathcal{R}^n$ is defined as
	\begin{align}
		\mathcal{Z} =\left\{ G\xi +c: A\xi =b, \xi \in \mathcal{C} \right\} , \nonumber
	\end{align}
	where $G\in \mathcal{R} ^{n\times n_g}$, $c\in \mathcal{R} ^n$, $A\in \mathcal{R} ^{n_c\times n_g}$, $b\in \mathcal{R} ^{n_c}$ and $\mathcal{C} =\left\{ \mathcal{C} _1,\mathcal{C} _2,\dots ,\mathcal{C} _{n_p} \right\} $, and $\mathcal{C} _i\subset \mathcal{R} ^{m_i}$ are convex sets with $\sum_{i=1}^{n_p}{m_i}=n_g$. 
\end{defn}
The CCG encompasses a wide range of useful set representations, including zonotopes, ellipsoids, and intervals~\cite{silvestre2021constrained}. Moreover, common set operations such as the Minkowski sum and intersection admit analytical expressions, enabling its application to approximate the inference attack.
\begin{prop}\cite{silvestre2021constrained}\label{Pro.CCGCompu}
	Given CCGs $\mathcal{X}\! =\!\left( G_x,c_x,A_x,b_x,\mathcal{C}_x \right) \!\subset\! \mathcal{R}^n$ and $\mathcal{Y}\! =\!\left( G_y,c_y,A_y,b_y,\mathcal{C}_y \right) \!\subset\! \mathcal{R}^n$, and a matrix $R\in\mathcal{R}^{m\times n}$, we have
	\scriptsize
	\begin{align}
		&R\mathcal{X} =\left( RG_x,Rc_x,A_x,b_x,\mathcal{C} _x \right) , \nonumber\\ 
		\mathcal{X} \!\oplus\! \mathcal{Y} &\!=\!\left(\! \left[ \begin{matrix}
			G_x&	\!	G_y\\
		\end{matrix} \right] \!, c_x\!+\!c_y, \mathrm{diag}\left(\! \left[ \begin{matrix}
			A_x&	\!	A_y\\
		\end{matrix} \right] \!\right) \!,\left[\! \begin{array}{c}
			b_x\\
			b_y\\
		\end{array} \!\right] \!, \left\{ \mathcal{C} _x,\mathcal{C} _y \right\} \!\!\right) \!, \nonumber \\ 
		\mathcal{X} \cap \mathcal{Y} &=\left( \left[ \begin{matrix}
			G_x&		0\\
		\end{matrix} \right] , c_x, \left[ \begin{matrix}
			A_x&		0\\
			0&		A_y\\
			G_x&		-G_y\\
		\end{matrix} \right] , \left[ \begin{array}{c}
			b_x\\
			b_y\\
			c_y-c_x\\
		\end{array} \right] , \left\{ \mathcal{C} _x,\mathcal{C} _y \right\} \right) . \nonumber
	\end{align}
\end{prop}
\normalsize
According to the computation rules in Proposition~\ref{Pro.CCGCompu}, the inference attack from \eqref{Eq.exactYEst0} to \eqref{Eq.exactYEst6} can be directly implemented, if we assume that the uncertainty sets $\mathcal{W}^x_k$, $\mathcal{W}^y_k$, $\mathcal{X}_{0|-1}$, and $\mathcal{Y}_{0|-1}$ are represented as CCGs. However, as shown in the next lemma, the computational complexity of CCG-based inference grows exponentially over time.}
\textcolor{blue}{\begin{prop}
	The computational complexity of the CCG-based inference attack at time $k$ is at least $\mathcal{O}\!\left(c^{k-1}n^3\right)$ for some constant $c>1$, and both the column dimension of the generator matrices $G$ and the number of constraints grow exponentially with $k$.
\end{prop}}
\begin{proof}
	\textcolor{blue}{The dominant operation in the inference attack from \eqref{Eq.exactYEst0} to \eqref{Eq.exactYEst6} is the multiplication of an $n\times n$ matrix with an $n\times m$ matrix, which has computational complexity $\mathcal{O}(mn^2)$.  
	For simplicity, we assume that at time $k$ the sets $\mathcal{X}_{k-1|k-1}$, $\mathcal{Y}_{k-1|k-1}$, $\mathcal{W}^x_{k}$, $\mathcal{W}^y_{k}$, and $\mathcal{M}^x_{k|k}$ all have generator matrices $G$ of size $n\times n$ and are described by $n$ constraints.  }
	
	\textcolor{blue}{According to Proposition~\ref{Pro.CCGCompu}, after one inference step, the Minkowski sum and intersection operations cause the generator matrix in $\mathcal{X}_{k|k}$ to grow to size $n\times (c\,n)$ for some constant $c>1$, while the number of constraints increases by a factor $d>1$.  
	Thus, both the column dimension of $G$ and the number of constraints grow exponentially with~$k$.  
	Consequently, due to the exponentially increasing column dimension, the computational complexity of matrix multiplications in the inference attack at time~$k$ is at least $\mathcal{O}(c^{k-1}n^3)$.}
\end{proof}
\textcolor{blue}{The high computational complexity of CCG-based inference renders real-time implementation of the inference attack and the privacy filter design in Sec.~\ref{Sec.PolicyDesign} intractable for large~$k$.  
Order-reduction techniques can be employed to reduce this complexity, albeit at the cost of some loss in inference accuracy.  
However, such techniques also complicate the analysis of the proposed volumetric privacy metric.  
For clarity and focus, we defer a detailed discussion of these techniques to future work.}
\subsection{Inference Attack Approximation via Interval Analysis}\label{Sec.InfInt}
\textcolor{blue}{We next consider an interval-based approximation approach for computing the inference sets. This approach can be viewed as a special case of the CCG-based inference, but it significantly reduces both computational and analytical complexity due to the efficiency of interval arithmetic.
\begin{defn}[Interval Representation]
	An interval $\mathcal{X} = \{ X \mid \underline{X} \leq X \leq \overline{X} \}$ is equivalently represented as $\mathcal{X} =\left[ \begin{matrix}
		\underline{X}^{\top}&		\overline{X}^{\top}\\
	\end{matrix} \right]^{\top} $,	where $\underline{X}$ and $\overline{X}$ are the lower and upper bounds, respectively. Equivalently, an interval can be expressed as a special case of the CCG representation,
	\begin{equation}
		\mathcal{X} = \left\{ \mathrm{diag}(p^x) \,\xi + c^x \;:\; \xi \in \mathbb{R}^{n_x},\; \|\xi\|_\infty \leq 1 \right\}, \nonumber
	\end{equation}
	with center $c^x = (\overline{X} + \underline{X})/2$ and radius $p^x = (\overline{X} - \underline{X})/2$. The volume of $\mathcal{X}$ is given by $
		\mathrm{Vol}(\mathcal{X}) = \prod_{i=1}^{n_x} \big(\overline{X}(i) - \underline{X}(i)\big) = \prod_{i=1}^{n_x} 2p^x(i)$, where $\overline{X}(i)$ and $\underline{X}(i)$ denote the upper and lower bounds of the $i$-th dimension, respectively.
\end{defn}}

\textcolor{blue}{To reduce computational complexity, a surrogate measure for the size of $\mathcal{X}$ can be defined as the total length across all dimensions $	\overline{\mathrm{Vol}}(\mathcal{X}) = \sum_{i=1}^{n_x} \big(\overline{X}(i) - \underline{X}(i)\big) = \sum_{i=1}^{n_x} 2p^x(i)$,
which is linear in the upper and lower bounds and therefore easier to compute. By the inequality of arithmetic and geometric means, this surrogate measure bounds the geometric volume, i.e.,  $\overline{\mathrm{Vol}}(\mathcal{X})/n \geq \sqrt[n]{\mathrm{Vol}(\mathcal{X})}$. Consequently, maintaining a large $\mathcal{X}$ implies a correspondingly large $\overline{\mathrm{Vol}}(\mathcal{X})$, making it a suitable surrogate metric.}

\textcolor{blue}{The following operations on intervals are defined consistently with this representation. Given a block matrix $A = [A_1,\,A_2]$, multiplication with an interval is$A \mathcal{X} = A_1 \underline{X} + A_2 \overline{X}$. For intervals $\mathcal{X}$ and $\mathcal{Y}$, their Minkowski sum is  $\mathcal{X} \oplus \mathcal{Y} = 
\begin{bmatrix}
	\underline{X} + \underline{Y} \\[0.3em]
	\overline{X} + \overline{Y}
\end{bmatrix}$, and their difference, used only for volume evaluation, is $
\mathcal{X} \setminus \mathcal{Y} =
\begin{bmatrix}
	\underline{X} - \underline{Y} \\[0.3em]
	\overline{X} - \overline{Y}
\end{bmatrix}$.
The intersection of $\mathcal{X}$ and $\mathcal{Y}$ is $\mathcal{X} \cap \mathcal{Y} =
\begin{bmatrix}
	\max\{\underline{X},\underline{Y}\} \\[0.3em]
	\min\{\overline{X},\overline{Y}\}
\end{bmatrix}$. These operations are considerably simpler to compute than the corresponding operations for CCGs described in Proposition~\ref{Pro.CCGCompu}.}

\textcolor{blue}{We now assume that the uncertainty sets $\mathcal{W}^x_k$, $\mathcal{W}^y_k$, $\mathcal{X}_{0|-1}$, and $\mathcal{Y}_{0|-1}$ are represented as intervals.  }
Under this assumption, the interval-based inference attack can be implemented using the following lemma.
\begin{lem} \label{Lm.tightestInferY}
	The recursive inference interval from \eqref{Eq.exactYEst0} to \eqref{Eq.exactYEst3} can be computed via
	\begin{align}
		\mathcal{M} _{k-1|k}^{x}&=\Psi \left( A_{1}^{-1} \right) \mathcal{M}^x_{k|k} \oplus  \Psi \left( -A_{1}^{-1}A_2 \right) \mathcal{Y}_{k-1|k-1} \nonumber\\&\oplus \Psi \left( -A_{1}^{-1}B_1 \right) \mathcal{W}_{k|k}^x  , \label{Eq.TightInterval1} \\
		\mathcal{M} _{k-1|k}^{y}&=\Psi \left( A_{2}^{-1} \right) \mathcal{M}^x_{k|k} \oplus \Psi \left( -A_{2}^{-1}A_1 \right) \mathcal{X}_{k-1|k-1}\nonumber\\&\oplus \Psi \left( -A_{2}^{-1}B_1 \right) \mathcal{W}_{k|k}^x   , \label{Eq.TightInterval2}\\
		\mathcal{X} _{k-1|k}&=\left[ \begin{array}{c}
			\max \left\{ \underline{M}_{k-1|k}^{x},\underline{X}_{k-1|k-1} \right\}\\
			\min \left\{ \overline{M}_{k-1|k}^{x},\overline{X}_{k-1|k-1} \right\}\\
		\end{array} \right] , \label{Eq.TightInterval3}\\
		\mathcal{Y} _{k-1|k}&=\left[ \begin{array}{c}
			\max \left\{ \underline{M}_{k-1|k}^{y},\underline{Y}_{k-1|k-1} \right\}\\
			\min \left\{ \overline{M}_{k-1|k}^{y},\overline{Y}_{k-1|k-1} \right\}\\
		\end{array} \right]  , \label{Eq.TightInterval4}
	\end{align}
	\begin{align}
		\!\!\!\!\!\!\mathcal{M} _{k|k-1}^{x}\!=\!\Psi\!& \left( A_1 \right)\! \mathcal{X} _{k-1|k}\!\oplus \!\Psi\! \left( A_2 \right)\! \mathcal{Y} _{k-1|k}\!\oplus \!\Psi\! \left( B_1 \right)\! \mathcal{W} _{k}^{x}, \label{Eq.TightIntervalM} \\
		\mathcal{X} _{k|k}&=\left[ \begin{array}{c}
			\max \left\{ \underline{M}_{k|k}^{x},\underline{M}_{k|k-1}^{x} \right\}\\
			\min \left\{ \overline{M}_{k|k}^{x},\overline{M}_{k|k-1}^{x} \right\}\\
		\end{array} \right], \\
		\!\!\!\!\!\!\!\!\mathcal{Y} _{k|k}\!=\!\Psi\! & \left( A_3 \right) \!\mathcal{X} _{k-1|k} \!\oplus\! \Psi\! \left( A_4 \right) \mathcal{Y} _{k-1|k} \!\oplus\! \Psi \! \left( B_2 \right) \mathcal{W}_{k}^{y}  \!,\! \label{Eq.TightInterval5} 
	\end{align}
	with 
	\begin{align}
		\Psi \left( \star \right) =\left[ \begin{matrix}
			\frac{\star+\left| \star \right|}{2}&		\frac{\star-\left| \star \right|}{2}\\
			\frac{\star-\left| \star \right|}{2}&		\frac{\star+\left| \star \right|}{2}\\
		\end{matrix} \right]. \nonumber
	\end{align}
	Also, the prior inference set of $Y_k$ is
	\begin{align} \label{Eq.TightInterval6}
		\!\!\!\!\!\mathcal{Y} _{k|k\!-\!1}\!=\!\Psi\! \left( \!A_3 \!\right)\! \mathcal{X}_{k-1|k-1} \!\oplus \!\Psi\! \left(\! A_4 \!\right) \!\mathcal{Y}_{k-1|k-1} \!\oplus \!\Psi\! \left( \!B_2 \!\right) \!\mathcal{W}_{k}^y  \!,\!\!
	\end{align}
	if $k\geqslant1$. If $k=0$, then $\mathcal{Y}_{0|0}=\mathcal{Y}_{0|-1}$ and
	\begin{align}\label{Eq.TightIntervalX0}
		\mathcal{X} _{0|0}=\left[ \begin{array}{c}
			\max \left\{ \underline{M}_{0|0}^{x},\underline{X}_{0|-1} \right\}\\
			\min \left\{ \overline{M}_{0|0}^{x},\overline{X}_{0|-1} \right\}\\
		\end{array} \right].
	\end{align}
\end{lem}
\begin{proof}
	See Appendix \ref{App.Lm.tightestInferY}.
\end{proof}
\textcolor{blue}{Given the interval-based inference approach described in Lemma~\ref{Lm.tightestInferY}, the computational complexity of the inference attack can be characterized as follows.
\begin{prop}
	The computational complexity of the inference attack via interval analysis is $\mathcal{O}(n^3)$.
\end{prop}
\begin{proof}
	The dominant operation in the inference attack involves matrix multiplication.  
	Since the matrix $\Psi(\star)$ has dimensions $(2n \times 2n)$, the corresponding computational complexity is $\mathcal{O}(n^3)$.
\end{proof}
Although the matrix multiplication with large $n$ can still be computationally demanding, the complexity of interval-based inference is substantially lower and remains constant over time, in sharp contrast to the exponentially growing complexity of CCG-based inference.}
\subsection{Properties of the Interval Inference Attack}
\textcolor{blue}{The inference attack exhibits several key properties. In particular, the radius of the uncertainty set $\mathcal{Y}_{k|k}$, $p_{k|k}^y = \overline{Y}_{k|k} - \underline{Y}_{k|k}$, is bounded by a function of the radius of the disturbance and the observation set, as formalized below.}
\begin{lem}\label{Lm.radiusY}
	For any $k \geq 1$, the radius of $\mathcal{Y}_{k|k}$ satisfies
	\begin{align}\label{Eq.radiusYInequality}
		p_{k|k}^{y} &\leq \Big( |A_3| + |A_4| |A_2^{-1}| + |A_4| |A_2^{-1} A_1| \Big) \overline{p}^x \nonumber \\
		&\quad + |A_4| |A_2^{-1} B_1| p_k^{w,x} + |B_2| p_k^{w,y},
	\end{align}
	where $\overline{p}^x \geq p^{m,x}_{j|j}$ for any $j \geq 0$, $p^{m,x}_{k|k}$, $p^{w,x}_k$, and $p^{w,y}_k$ are the radii of $\mathcal{M}_{k|k}^x$, $\mathcal{W}_k^x$, and $\mathcal{W}_k^y$, respectively, and $|A|$ denotes the matrix with elementwise absolute values, i.e., $|A| = [|a_{i,j}|]$.
\end{lem}
\begin{proof}
	See Appendix \ref{App.Lm.radiusY}.
\end{proof}

\textcolor{blue}{Since the volume of $\mathcal{Y}_{k|k}$ is given by the product of $2 p_{k|k}^y$ over all dimensions, $\mathrm{Vol}(\mathcal{Y}_{k|k})$ is similarly bounded by a function of the radius of the observation set $\mathcal{M}_{k|k}^x$. Consequently, a small $\mathcal{M}_{k|k}^x$ implies low privacy, and the adversary retains limited uncertainty after performing the inference attack.}

\textcolor{blue}{Furthermore, by comparing the predicted and posterior uncertainty sets, as in \eqref{Eq.calY1} and \eqref{Eq.exactYEst4}, the reduction of uncertainty can be quantified using the surrogate measure $\overline{\mathrm{Vol}}(\Delta \mathcal{Y}_{k|k})$, where
\begin{align}
	\Delta \mathcal{Y}_{k|k} = \mathcal{Y}_{k|k-1} \setminus \mathcal{Y}_{k|k}.
\end{align}
Since $\mathcal{Y}_{k|k} \subseteq \mathcal{Y}_{k|k-1}$, the surrogate volume can be computed as
\begin{align}\label{Eq.diffVolume}
	\overline{\mathrm{Vol}}(\Delta \mathcal{Y}_{k|k}) 
	= \overline{\mathrm{Vol}}(\mathcal{Y}_{k|k-1}) - \overline{\mathrm{Vol}}(\mathcal{Y}_{k|k}).
\end{align}
Thus, increasing the privacy level $\overline{\mathrm{Vol}}(\mathcal{Y}_{k|k})$ is equivalent to reducing the amount of uncertainty reduction $\overline{\mathrm{Vol}}(\Delta \mathcal{Y}_{k|k})$, since the prior uncertainty $\overline{\mathrm{Vol}}(\mathcal{Y}_{k|k-1})$ is fixed at time $k$. As shown in the next theorem, the amount of uncertainty reduction, is bounded by the new information extracted from $\mathcal{M}^x_{k|k}$.}
\begin{thm}\label{Th.PriLeakBounds}
The amount of uncertainty reduction at $k$ is
\begin{align}\label{Eq.PrivacyLeaksMultiDimen}
	\!\!\!\!	\overline{\mathrm{Vol}}\left( \Delta \mathcal{Y} _{k|k} \right)\! =\!\left\| \Psi\! \left( A_3 \right) \!\Delta \mathcal{X} _{k-1|k} \oplus \Psi\! \left( A_4 \right) \!\Delta \mathcal{Y} _{k-1|k} \right\|_1 \!, \!\!
\end{align}
with bounds
\begin{align}
	\overline{\mathrm{Vol}}&\left( \Delta \mathcal{Y} _{k|k} \right) \geqslant 2\left\|c_{k|k}^{y}-c_{k|k-1}^{y} \right\|_1, \nonumber \\
	\!\overline{\mathrm{Vol}}\!\left( \Delta \mathcal{Y} _{k|k} \right)\! \!\leqslant & \!\left\|\! A_3 \!\right\|_1\!\overline{\mathrm{Vol}}\!\left(\! \Delta \mathcal{X} _{k-1|k} \right)\! +\!\left\| \!A_4\! \right\|_1\!\overline{\mathrm{Vol}}\!\left(\! \Delta \mathcal{Y} _{k-1|k} \right) \!, \nonumber
\end{align}
where
$\Delta\mathcal{X}_{k-1|k} \!\!=\!\!\mathcal{X} _{k-1|k-1}\!\setminus\!\mathcal{X} _{k-1|k}$, $\Delta\! \mathcal{Y}_{k-1|k} \!\!=\!\!\mathcal{Y} _{k-1|k-1}\!\setminus\!\mathcal{Y} _{k-1|k}$, and $\|c_{k|k}^y - c_{k|k-1}^y\|_1$ quantifies the change in the central estimate due to the observation $\mathcal{X}_{k|k}$.
\end{thm}

\begin{proof}
See Appendix \ref{App.Th.PriLeakBounds}.
\end{proof}

Combining \eqref{Eq.diffVolume} and Theorem \ref{Th.PriLeakBounds}, the privacy level $\overline{\mathrm{Vol}}(\mathcal{Y}_{k|k})$ can be bounded as follows.

\begin{lem}\label{Lm.PrivBound}
	The privacy level satisfies
	\begin{align}
		&\overline{\mathrm{Vol}}\!\left( \!\mathcal{Y} _{k|k-1}\! \right) \!-\!\left\| \!A_3\! \right\|_1 \!\overline{\mathrm{Vol}}\!\left( \!\Delta \mathcal{X} _{k-1|k} \right) \!-\!\left\| \!A_4\! \right\|_1 \!\overline{\mathrm{Vol}}\!\left( \!\Delta \mathcal{Y} _{k-1|k} \right)  \nonumber
		\\
		&\leqslant  \overline{\mathrm{Vol}}\!\left( \!\mathcal{Y} _{k|k}\! \right) \!\leqslant \overline{\mathrm{Vol}}\left( \mathcal{Y} _{k|k-1} \right) -2\left\| c_{k|k}^{y}-c_{k|k-1}^{y} \right\|_1 . \nonumber
	\end{align}
\end{lem}
\textcolor{blue}{As a result, we can reduce the extracted information $\overline{\mathrm{Vol}}\left( \Delta \mathcal{X} _{k-1|k} \right)$ and $\overline{\mathrm{Vol}}\left( \Delta \mathcal{Y} _{k-1|k} \right)$ to increase the privacy level $\overline{\mathrm{Vol}}\left( \mathcal{Y}_{k|k}\right)$ via designing proper observation set $\mathcal{M}^x_{k|k}$. Moreover, if the privacy level is high, the adversary’s ability to update its central estimate $c_{k|k}^y$ is also limited, as indicated by the small value of $\left\| c_{k|k}^y - c_{k|k-1}^y \right\|_1$. Based on this observation, $\mathcal{M}^x_{k|k}$ can also be designed to hinder accurate central estimate updates, further improving the privacy level.}
\section{Privacy Filter Design Problem Using the Volumetric Privacy measure}\label{Sec.PolicyDesign}
\textcolor{blue}{As discussed previously, an inappropriate choice of the observation set would cause privacy leakage of the private state through inference attacks. To mitigate this risk, we address the privacy filter design problem in this section. The proposed filter determines an appropriate observation set that achieves a desirable balance between preserving the data utility of the public state and ensuring the privacy protection of the private state.
\subsection{The Structure of Privacy Filter}
We begin by defining the decision domain of the privacy filter as follows. At time $k$, given the last decision set $\mathcal{X}_{k-1|k-1}$ and the private set $\mathcal{Y}_{k-1|k-1}$, the inference set of $X_k$ can be computed via,
\begin{align}
	\mathcal{Y} _{k|k-1}=A_3\mathcal{X} _{k-1|k-1}\oplus A_4\mathcal{Y} _{k-1|k-1}\oplus B_2\mathcal{W} _{k}^{y}, \nonumber
\end{align}
which contains all possible public states that can be reached from any states in $\mathcal{X}_{k-1|k-1}$ and $\mathcal{Y}_{k-1|k-1}$. Therefore, $\mathcal{X}_{k|k-1}$ is the maximum observation set $\mathcal{M}^x_{k|k}$ that the filter can release, i.e., $\mathcal{M}^x_{k|k}\subseteq \mathcal{X}_{k|k-1}$. To maintain high data utility, the uncertainty set $\mathcal{X}_{k|k}$ must satisfy the following constraint:
\begin{equation}
	\mathrm{Vol}\left(\mathcal{X}_{k|k}\right) \leq \epsilon^x, \nonumber
\end{equation}
where $\epsilon^x > 0$ specifies the desired upper bound on the uncertainty volume. Since $\mathcal{X}_{k|k}$ is a subset of the observation set $\mathcal{M}_{k|k}^x$, this constraint can be equivalently enforced on the larger set, $\mathrm{Vol}\left(\mathcal{M}_{k|k}^x\right) \leq \epsilon^x$, which simplifies the design of the observation set while ensuring that the utility requirement is satisfied.}

\textcolor{blue}{To reduce privacy leakage while preserving data utility, we design the privacy filter illustrated in Fig. \ref{Fig.PrivacyFilterStructure}. The design consists of two steps:  
(1) randomly generate a set $\mathcal{S}_{k|k}^x$ such that $\mathcal{S}_{k|k}^x \subseteq \mathcal{X}_{k|k-1}$ and $\mathrm{Vol}(\mathcal{S}_{k|k}^x) \leq \epsilon^x$;  
(2) optimize the observation set $\mathcal{M}^x_{k|k}$, which contains $\mathcal{S}_{k|k}^x$, to maximize the privacy level. Specifically, in the optimization step, given $\mathcal{S}_{k|k}^x$, we maximize the privacy level under the inference attack \eqref{Eq.exactYEst0}-\eqref{Eq.exactYEst6} by solving
\begin{align}\label{Eq.OptP1}
	\mathbf{P_1}:~ &\max_{\mathcal{M}^x_{k|k}}~ \mathrm{Vol}\!\left( \mathcal{Y}_{k|k} \right) \\
	\text{s.t.}~ &\left\{
	\begin{array}{l}
		\mathcal{S}^x_{k|k} \subseteq \mathcal{M}^x_{k|k}, \\[0.2em]
		\mathcal{M}^x_{k|k} \subseteq \mathcal{X}_{k|k-1}, \\[0.2em]
		\mathrm{Vol}\!\left( \mathcal{M}^x_{k|k} \right) \leq \epsilon^x, \\[0.2em]
		\eqref{Eq.exactYEst0} \text{-} \eqref{Eq.exactYEst6}.
	\end{array}
	\right. .
\end{align}
Note that $\mathcal{S}^x_{k|k}$ is randomly generated as a subset of $\mathcal{X}_{k|k-1}$ and can be made sufficiently small in practice. For instance, the random set $\mathcal{S}^x_{k|k}$ may contain only the true public state $x_k$. According to Lemma \ref{Lm.radiusY}, a sufficiently small observation set could lead to potential privacy leakage. Therefore, recovering $\mathcal{S}^x_{k|k}$ from $\mathbf{P_1}$ must be avoided. In the following, we demonstrate that an attacker cannot recover $\mathcal{S}^x_{k|k}$ by inverting the optimization problem $\mathbf{P_1}$, owing to the randomization mechanism embedded in the privacy filter.
\begin{prop}\label{Pro.invertLP}
	The attacker cannot obtain the smaller set $\mathcal{S}^x_{k|k}$ by inverting the optimization problem $\mathbf{P_1}$.
\end{prop}
\begin{proof}
	First, $\mathcal{S}_{k|k}^x$ is selected as a random subset of $\mathcal{X}_{k|k-1}$ that contains the true state $x_k$. Consequently, $x_k$ may reside on the boundary of $\mathcal{S}_{k|k}^x$.  
	Next, let $\mathcal{M}_{k|k}^{x,\star}$ denote the optimal observation set. In some cases, $\mathcal{S}_{k|k}^x$ may coincide with $\mathcal{M}_{k|k}^{x,\star}$, in which case $x_k$ may also lie on the boundary of $\mathcal{M}_{k|k}^{x,\star}$.  
	Therefore, an attacker cannot reconstruct a strictly smaller feasible set containing $x_k$ by inverting the optimization process.
\end{proof}
As a result, the proposed privacy filter exhibits the following properties: (1) In the absence of inference attacks, the filter output $\mathcal{M}^x_{k|k}$ satisfies the utility constraint. (2) In the presence of the inference attack described in \eqref{Eq.exactYEst0}--\eqref{Eq.exactYEst6}, the filter output maximizes the privacy level. (3) The filter is robust against reverse attacks that attempt to recover the sufficiently small set $\mathcal{S}_{k|k}^x$, thereby reducing the risk of privacy compromise from adversaries exploiting structural vulnerabilities.}

\textcolor{blue}{Moreover, the computational complexity of both the inference attack in \eqref{Eq.exactYEst0}--\eqref{Eq.exactYEst6} and the volume computation increases with the complexity of the set representations. Consequently, a trade-off exists between the achievable privacy enhancement of the proposed filter and its computational cost. While more sophisticated set representations may improve the accuracy of privacy evaluation, for efficiency and clarity, we next present a concrete design based on the interval approximation described in Sec. \ref{Sec.InfInt}.}

\textcolor{blue}{In addition, as discussed in Sec.\ref{Sec.InfInt}, the surrogate measure $\overline{\mathrm{Vol}}(\cdot)$ bounds the volume of interval, and could simplify the computation complexity. We adopt the surrogate volumetric measure $\overline{\mathrm{Vol}}(\cdot)$ in the optimization to further reduce computational complexity, and we verify that using this surrogate measure still leads to improved privacy levels in Sec. \ref{Sec.NumericalVer}.}
\begin{figure}[!t]
	\centering
	\includegraphics[width=3.2in]{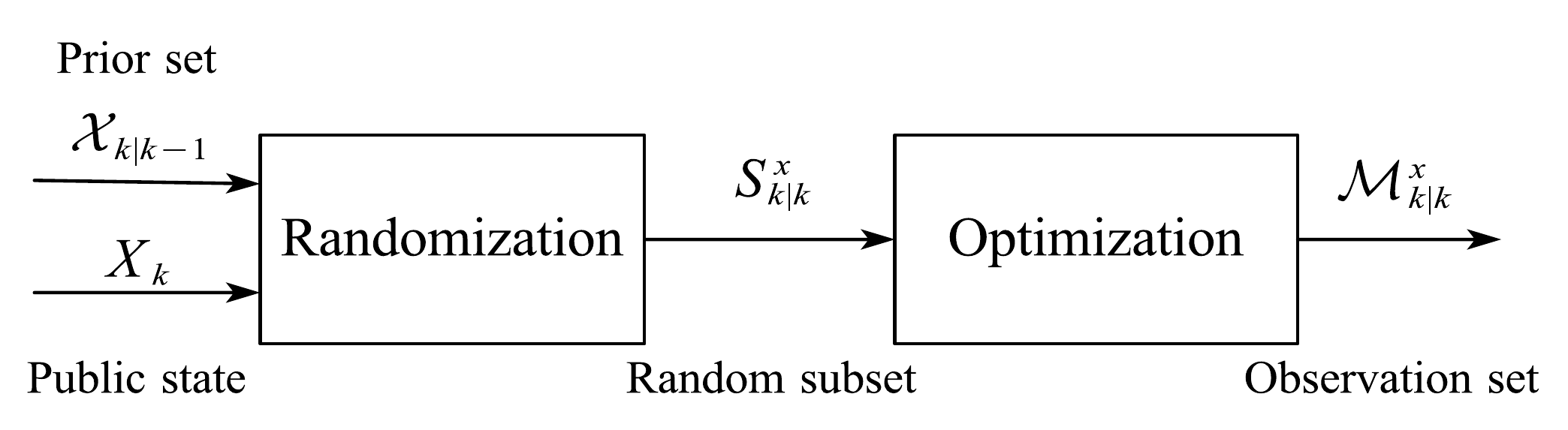}
	\caption{Structure of the proposed privacy filter.}
	\label{Fig.PrivacyFilterStructure}
\end{figure}
\subsection{Randomization}\label{Sec.randomPolicy}
We consider the following random set
\begin{align}\label{Eq.randIntervalPolicyLarge}
	{\mathcal{S}} _{k|k}^{x}=\left[ \begin{array}{c}
		x_k-{\alpha}_k \left( x_k-\underline{X}_{k|k-1} \right)\\
		x_k+{\beta}_k \left( \overline{X}_{k|k-1}-x_k \right)\\
	\end{array} \right] ,
\end{align}
where $\alpha_k$ and $\beta_k$ are uniform random variables with
\begin{align}
	{\alpha}_k\!\in\! \left[\!0, \frac{\epsilon^x}{2\left\|x_k-\underline{X}_{k|k-1}\right\|_1}\right]\!\!, 
	{\beta}_k \!\in\! \left[\!0, \frac{\epsilon^x}{2\left\|\overline{X}_{k|k-1}-x_k\right\|_1}\right]\!\!. \nonumber
\end{align}
Since $\left( x_k-\underline{X}_{k|k-1} \right)$ is the radius from the actual public state $X_k=x_k$ to the lower endpoint of $\mathcal{X}_{k|k-1}$, and $\left( \overline{X}_{k|k-1}-x_k \right)$ is the radius from $x_k$ to the upper endpoint of $\mathcal{X}_{k|k-1}$, the random set ${\mathcal{S}} _{k|k}^{x}$ becomes a subset of $\mathcal{X}_{k|k-1}$ that contains the actual public state. Also, we can shown $\mathcal{S}^x_{k|k}$ satisfies the utility constraint as follows,
\begin{align}
	\overline{\mathrm{Vol}}\left(\mathcal{S}^x_{k|k}\right) =& {\beta}_k\left\| \overline{X}_{k|k-1}-x_k  \right\|_1+{\alpha}_k\left\|  x_k-\underline{X}_{k|k-1} \right\|_1 \nonumber
	\\ \leqslant &  \frac{\epsilon^x}{2\left\|\overline{X}_{k|k-1}-x_k\right\|_1} \left\| \overline{X}_{k|k-1}-x_k  \right\|_1 \nonumber\\ &+ \frac{\epsilon^x}{2\left\|x_k-\underline{X}_{k|k-1}\right\|_1} \left\| x_k-\underline{X}_{k|k-1}  \right\|_1 \nonumber \\ = & \epsilon^x . \nonumber
\end{align}
We next restrict $\mathcal{S}^x_{k|k}$ be the subset of the observation set $\mathcal{M}^x_{k|k}$, and optimize $\mathcal{M}^x_{k|k}$ to improve the privacy level.
\subsection{Privacy Filter Optimization}
In this subsetion, we demonstrate that the optimization problem $\mathbf{P_1}$ based on the interval inference can be solved via linear programming.
\begin{thm}\label{Th.linearOpt}
	The privacy filter optimization problem $\mathbf{P_1}$ with the surrogate privacy measure $\overline{\mathrm{Vol}}\left(\mathcal{Y}_{k|k}\right)$ and utility measure $\overline{\mathrm{Vol}}\left(\mathcal{M}^x_{k|k}\right)$ is equivalent to the following linear programming 
	\begin{align}
		\mathbf{P_2}: &\max_{\epsilon ^y,\mathcal{M}^x_{k|k},p_{k-1|k}^{\Delta x} ,p_{k-1|k}^{\Delta y}} \epsilon ^y \nonumber \\
		&\left\{ \begin{array}{c}
			\left\| \left| A_3 \right|p_{k-1|k}^{\Delta x}+\left| A_4 \right|p_{k-1|k}^{\Delta y} \right\|_1\geqslant \epsilon ^y\\
			\left\| \overline{M}^x_{k|k}-\underline{M}^x_{k|k} \right\|_1\leqslant \epsilon ^x\\
			\underline{X}_{k|k-1}\leqslant \underline{M}^x_{k|k}\leqslant \underline{S}^x_{k|k} \\
			\overline{S}^x_{k|k}\leqslant \overline{M}^x_{k|k}\leqslant \overline{X}_{k|k-1}\\
			\eqref{Eq.TightInterval1}-\eqref{Eq.TightInterval2}
		\end{array} \right. , \nonumber \\
		&\left\{ \begin{array}{c}
			p_{k-1|k}^{\Delta z}\geqslant 0\\
			p_{k-1|k}^{\Delta z}\geqslant p_{k-1|k-1}^{z}-p_{k-1|k}^{m,z}\\
			2p_{k-1|k}^{\Delta z}\geqslant \overline{Z}_{k-1|k-1}-\overline{M}_{k-1|k}^{z}\\
			2p_{k-1|k}^{\Delta z}\geqslant \underline{M}_{k-1|k}^{z}-\underline{Z}_{k-1|k-1}\\
		\end{array} \right. , \label{Eq.optDeltaConstraint1} 
	\end{align}
	where $p_{k-1|k}^{\Delta x}\!\in \mathcal{R}^{n_x}$, $Z=X,Y$, $\epsilon^y\geqslant 0$ and $\mathcal{M}^x_{k|k}\subseteq \mathcal{R}^{2n_x}$.
\end{thm}
\begin{proof}
	See Appendix \ref{App.Th.linearOpt}
\end{proof}
\begin{figure*}
	\centering
	\includegraphics[width=6.6in]{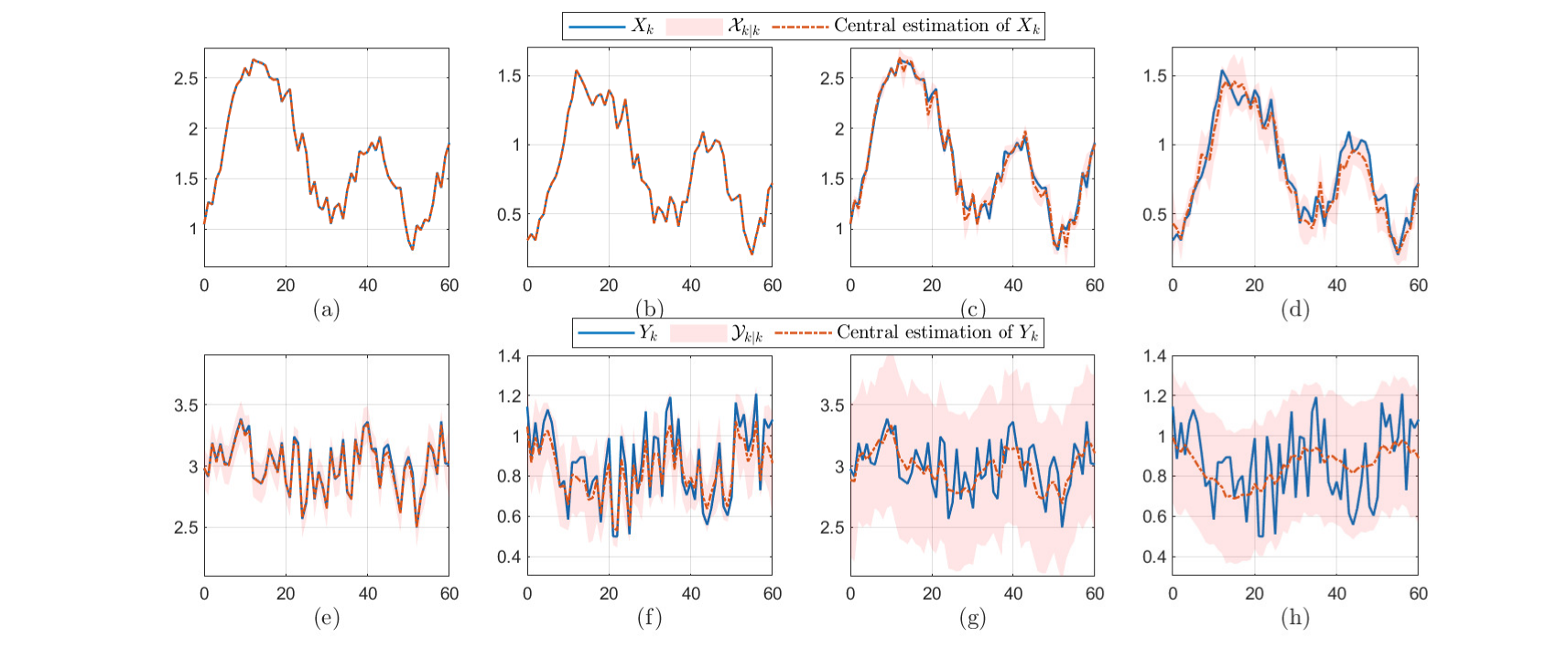}
	\caption{Inference attack results after applying the optimal privacy filter:  
		(a), (b) Estimated $X_k$ and (e), (f) Estimated $Y_k$ for $\mathrm{Vol}\!\left(\mathcal{M}_{k|k}^{x}\right) \leq 0.01$;  
		(c), (d) Estimated $X_k$ and (g), (h) Estimated $Y_k$ for $\mathrm{Vol}\!\left(\mathcal{M}_{k|k}^{x}\right) \leq 0.5$.}
	\label{Fig.PrivEstResXY}
\end{figure*}
Consequently, we can solve the linear programming problem $\mathbf{P_2}$ to obtain the optimal observation set that defends the system against the inference attack defined in Section \ref{Sec.InferAttack}. 
%\textcolor{blue}{\begin{rem}
%		For other approximation techniques, we can similarly design the privacy filter by randomization and optimization. For example, 
%		given a CCG $\mathcal{S}^x_{k|k} =\left( G^x_{k|k},c^x_{k|k},A^x_{k|k},b^x_{k|k},\mathcal{C}^x_{k|k} \right) \subset \mathcal{R}^n$, we can construct its subset $\mathcal{Y}$ by randomly select a subset of $\mathcal{C}^x_{k|k}$, i.e., $\mathcal{S}^x_{k|k} =\left( G^{s,x}_{k|k},c^{s,x}_{k|k},A^{s,x}_{k|k},b^{s,x}_{k|k},\mathcal{C}^{s,x}_{k|k} \right) $ with $\xi_k\in\mathcal{C}^{s,x}_{k|k} \subset \mathcal{C}^x_{k|k}$, where $\xi_k$ satisfies $x_k=G\xi_k+c$ and $A\xi_k=b$. Also, we restrict that $\mathrm{Vol}\left(\mathcal{S}^x_{k|k}\right) < \epsilon^x$. However, the computation complexity increases greatly.
%\end{rem}}
\section{Numerical Verification}\label{Sec.NumericalVer}
In this section, we study the performance of privacy  filter for the production-inventory problem with the following parameters
$$A_1=\left[ \begin{matrix}
		1.00&		0.00\\
		0.00&		1.00\\
	\end{matrix} \right] ,A_2=\left[ \begin{matrix}
		0.40&		0.80\\
		0.60&		0.20\\
	\end{matrix} \right], $$
$$A_3=\left[ \begin{matrix}
	0.50&		-0.90\\
	-0.10&		-0.10\\
\end{matrix} \right] ,A_4=\left[ \begin{matrix}
	-0.10&		-0.90\\
	0.10&		0.00\\
\end{matrix} \right] , $$
$$B_1=\left[ \begin{matrix}
	-1.00&		0.00\\
	0.00&		-1.00\\
\end{matrix} \right] ,B_2=\left[ \begin{matrix}
	4.20&		0.00\\
	0.00&		2.40\\
\end{matrix} \right], $$
$$\!\left(\! \mathcal{W} _{k}^{x}\! \right) ^\top\!\!=\!\left[\! \begin{matrix}
	1.74&		1.91&		1.94&		2.01\\
\end{matrix} \!\right]\!\!, \left(\! \mathcal{W} _{k}^y \!\right) ^\top\!\!=\!\left[\! \begin{matrix}
	0.91&		0.23&		0.95&		0.43\\
\end{matrix} \!\right]\!\!.$$
The initial state sets are assumed to be
\begin{align}
	\left( \mathcal{X} _{0|-1} \right) ^{\top}&=\left[ \begin{matrix}
		1.00&		0.24&		1.20&		0.40\\
	\end{matrix} \right],  \label{Eq.initialSet1} \\
	\left( \mathcal{Y} _{0|-1} \right) ^{\top}&=\left[ \begin{matrix}
		2.40&		0.60&		3.70&		1.30\\
	\end{matrix} \right].  \label{Eq.initialSet2}
\end{align}
In our simulation, the initial states are uniformly sampled from the bounded sets \eqref{Eq.initialSet1}-\eqref{Eq.initialSet2}. To simulate the approximate periodic fluctuations in demand and productivity, the actual disturbance are set to be
\begin{align}
	&\left( W _{k}^{x} \right) ^\top=\left[ \begin{matrix}
		1.88+0.03\cos\left(\frac{2\pi k}{30+7\rho_k}\right)	&	1.94\\
	\end{matrix} \right], \nonumber \\
	\!\!\left( W _{k}^y \right)\! ^\top \!&=\!\left[\! \begin{matrix}
		0.944\!+\!0.006\cos\left(\!\frac{2\pi k}{7+2\gamma_k}\!\right) \!\!	& \!\!	0.33+0.094\sin\left(\!\frac{2\pi k}{7+4\tau_k}\!\right)\\
	\end{matrix} \!\right]\!, \nonumber
\end{align}
where $\rho_k$, $\gamma_k$ and $\tau_k$ are uniform random variables in $\left[0,1\right]$. As discussed in Section \ref{Sec.ModelMeasure}, the production rate is private but the inventory information has to be released.

\textcolor{blue}{We first plot the trajectories of the system states and their corresponding interval tubes in Fig.~\ref{Fig.PrivEstResXY} under the optimal privacy filter design for different values of $\epsilon^x$. The shaded pink areas represent the interval tubes, i.e., the uncertainty sets of the system states, which quantify the range of values that an adversary can infer. As shown in Fig.~\ref{Fig.PrivEstResXY}, for $\epsilon^x = 0.01$, the adversary’s uncertainty about the private production rate is small. However, increasing $\epsilon^x$ to $0.5$ slightly reduces the utility of the inventory information but significantly enlarges the adversary's uncertainty. }

\textcolor{blue}{To further illustrate the effectiveness of the privacy filter, we consider one possible adversary estimate based on the central points of the posterior intervals. For the public state $\mathcal{X}_{k|k}$, the adversary’s central estimate is given by $\frac{\overline{X}_{k|k} + \underline{X}_{k|k}}{2}$. As shown in Fig.~\ref{Fig.PrivEstResXY}, when $\epsilon^x$ increases from $0.01$ to $0.5$, the adversary’s central estimate of the production rate becomes significantly less accurate, whereas the central estimate of the inventory information remains accurate. This observation numerically confirms Theorem~\ref{Th.PriLeakBounds}, showing that a higher privacy level prevents the adversary from refining an incorrect central estimate. Hence, the proposed privacy filter effectively mitigates the leakage of production rate information, while introducing a controlled loss of accuracy in the inventory information.}
\begin{figure}[h]
	\centering
	\includegraphics[width=2.8in]{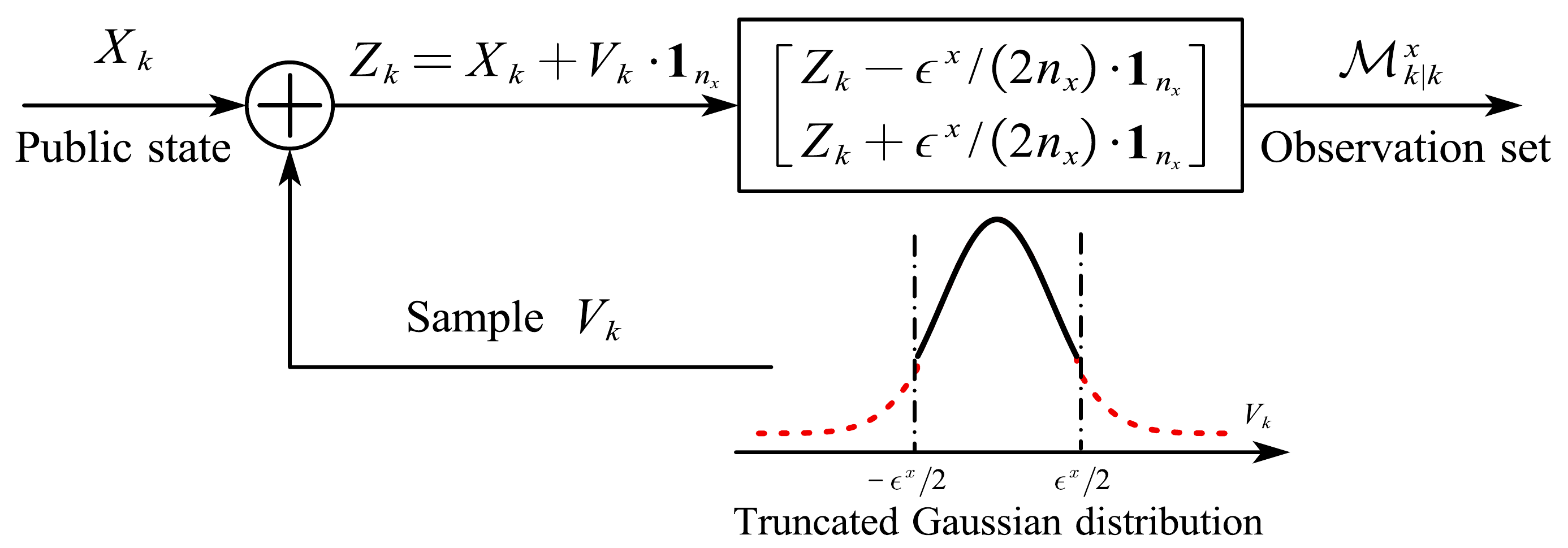}
	\caption{The truncated Gaussian mechanism. }
	\label{Fig.additiveNoiseApp}
\end{figure}

\textcolor{blue}{We also evaluate the utility--privacy trade-off achieved by the proposed optimal privacy-filtering policy and compare it with two benchmark mechanisms: the noiseless quantization method presented in \cite{farokhi2021noiseless} and the truncated Gaussian mechanism for differential privacy introduced in \cite{chen2024bounded}. In the quantization-based approach, the state $x_k$ is processed through a static quantizer that satisfies the utility constraint, and the quantization bin containing $x_k$ is publicly released as $\mathcal{M}_{k|k}^x$. In contrast, the truncated Gaussian mechanism, illustrated in Fig.~\ref{Fig.additiveNoiseApp}, perturbs the original state $x_k$ with additive noise $v_k$ drawn from a zero-mean truncated Gaussian distribution supported on the interval $\left[-\epsilon^x/2,\, \epsilon^x/2\right]$ and having variance $\left(\epsilon^x\right)^2$. The perturbed observation set is then released in the form
\begin{equation}
	\mathcal{M}_{k|k}^x =
	\left[
	\begin{array}{c}
		z_k - \frac{\epsilon^x}{2 n_x} \cdot \mathbf{1}_{n_x} \\
		z_k + \frac{\epsilon^x}{2 n_x} \cdot \mathbf{1}_{n_x}
	\end{array}
	\right], \nonumber
\end{equation}
where $n_x$ denotes the dimension of $x_k$ and $\mathbf{1}_{n_x}$ is the all-ones vector of length $n_x$. Because the additive noise $v_k$ is bounded within $\left[-\epsilon^x/2,\, \epsilon^x/2\right]$, the publicly released state is guaranteed to lie within the interval $\mathcal{M}_{k|k}^x$. }

%\textcolor{blue}{Note that the quantization mechanism and the truncated Gaussian mechanism assume time-invariant domains of states, while the sets of states in systems subject to disturbance can be estimated and described more precise as addressed by the inference attack. Differently, the proposed volumetric approach estimates the time-varying private state set and adjusts the output accordingly.}

\textcolor{blue}{We evaluate the privacy-utility trade-off by plotting the average privacy level of the production rate against the average utility of the inventory in Fig. \ref{Fig.TradeoffEvaInt}. For a fair and clear comparison, the privacy level and utility values for the truncated Gaussian mechanism are normalized to the range $[0,1]$, and the same scaling parameters are applied to the other two mechanisms. The results demonstrate that an increase in inventory utility corresponds to a reduction in the privacy level of the production rate, thereby confirming the intrinsic trade-off between data utility and privacy protection.}

\textcolor{blue}{As discussed in \cite{chen2024bounded}, if the interval length, i.e., volume, of domain satisfies certain conditions, the truncated Gaussian mechanism ensures differential privacy. However, given the volume constraint, the shape of the public state set can be arbitrary, and certain shapes may lead to substantial volumetric leakage of the private state after inference attacks. The proposed volumetric method explicitly accounts for this by considering the geometry of the set based on the assumed inference attack, not only its volume. Therefore, it achieves higher privacy levels while maintaining lower data distortion compared with the two static mechanisms.}

\textcolor{blue}{It is worth noting that an adversary could employ more sophisticated estimation techniques, \emph{e.g.}, CCG-based approximation, to infer the private set more accurately from the interval observations provided by the privacy filter. Nevertheless, as shown in Fig. \ref{Fig.TradeoffEvaCCG}, the proposed privacy filter still outperforms the other two mechanisms, leveraging knowledge of the underlying state evolution to reduce conservativeness.}

\begin{figure}
	\centering
	\includegraphics[width=2.8in]{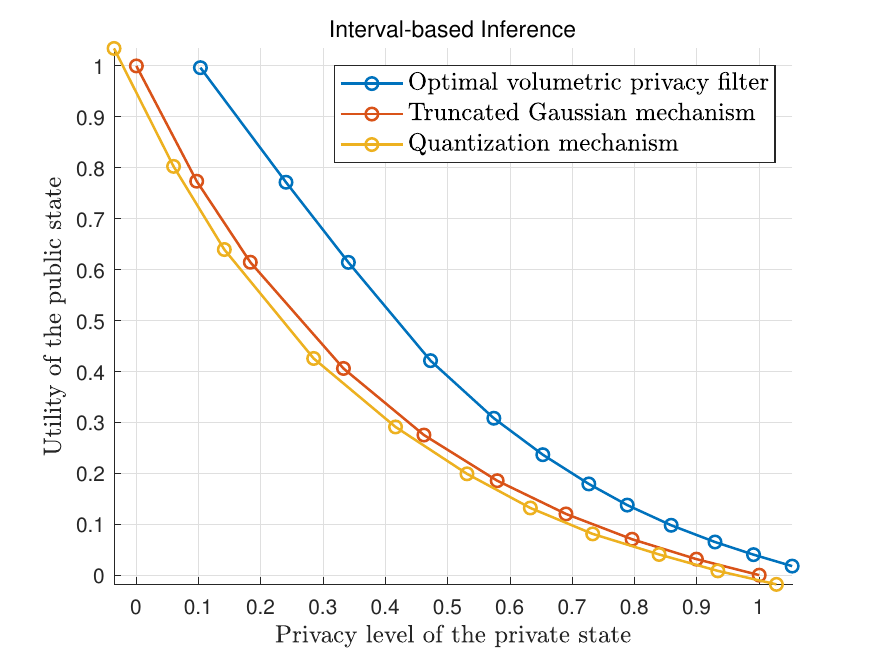}
	\caption{Interval-based inference given the interval privacy filter: the privacy level of the private state and utility of the public state.}
	\label{Fig.TradeoffEvaInt}
\end{figure}
\begin{figure}
	\centering
	\includegraphics[width=2.8in]{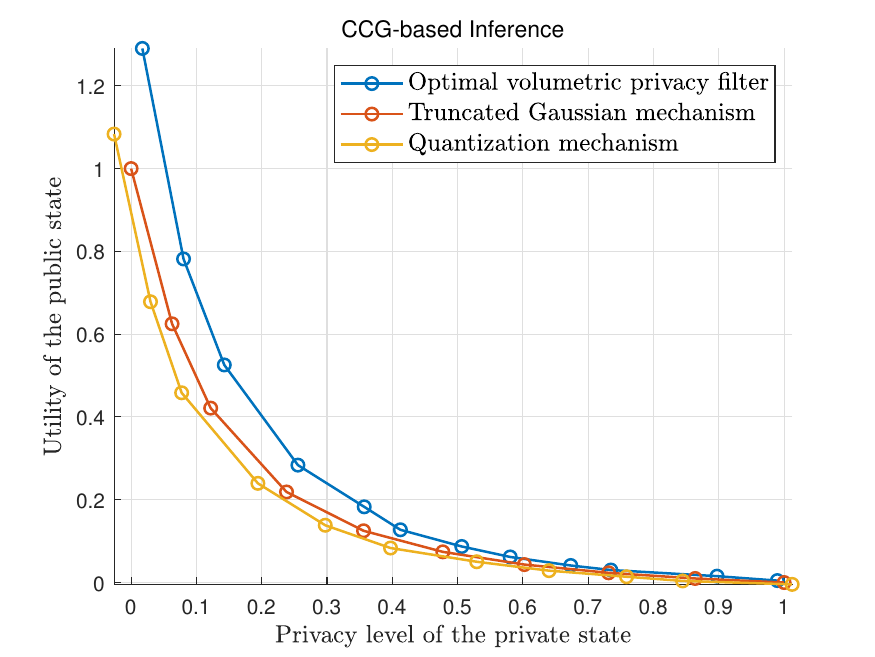}
	\caption{CCG-based inference given the interval privacy filter: the privacy level of the private state and utility of the public state.}
	\label{Fig.TradeoffEvaCCG}
\end{figure}
\section{Conclusion}\label{Sec.Conclusion}
\textcolor{blue}{In this paper, we develop a volumetric framework for privacy analysis and defense in dynamic systems subject to bounded disturbance. An inference attack, whereby an adversary estimates the private information, is formalized, and a volumetric measure is introduced to quantify the resulting privacy level. We develop computational methods based on interval analysis, and establish the theoretical properties of the measure. Furthermore, we propose an optimization-based approach for privacy filter design to defend the system against inference attacks. The effectiveness of our method is demonstrated through a production-inventory case study. }

\textcolor{blue}{It is noted that the performance of the volumetric privacy measure inherently depends on the selected set-membership estimation techniques, and its evaluation accuracy varies with different set representations. Future research will focus on developing approximation methods that ensure improved accuracy, robustness and broader applicability.}

\appendix
\section{Proof of Lemma \ref{Lm.tightestInferY}}\label{App.Lm.tightestInferY}
At the time step $k=0$, the adversary only has prior knowledge, i.e., $\mathcal{Y}_{0|-1}$, therefore, its inference set is $\mathcal{Y}_{0|0} = \mathcal{Y}_{0|-1}$. Also, since at the time step $k=0$, the backward calibration \eqref{Eq.exactYEst0} and \eqref{Eq.exactYEst1} is not available, the adversary can only calibrate the public state set with its prior knowledge $\mathcal{X}_{0|-1}$ and the observation set $\mathcal{M}_{0|0}^x$ according to \eqref{Eq.TightIntervalX0}.

To prove Lemma \ref{Lm.tightestInferY} for $k\geqslant1$, we need the following lemma that computes the tightest interval by forward reachability analysis.
\begin{lem} \label{Lm.basicInterval}
	\cite{bako2019interval,bako2022interval} Consider the static system $S=AM+BW$, where $M$ and $W$ are bounded intervals, the tightest interval for $S$, i.e., $\mathcal{S}$ can be computed as
	\begin{align}
		\mathcal{S} =\Psi \left( A \right) \mathcal{M} \oplus \Psi \left( B \right) \mathcal{W}. \nonumber
	\end{align}
	Also, we can compute its radius and center via
	\begin{align}
		p^{s}=\left| A \right|p^{m}&+ \left| B \right| p^w, \nonumber \\ 
		c^s=Ac^m&+Bc^w. \nonumber
	\end{align}
\end{lem}
%When $A_1$ and $A_2$ are invertible, we have 
%\begin{align}
%	X_{k-1}&=A_{1}^{-1}X_k+\left( -A_{1}^{-1}A_2 \right) Y_{k-1}+\left( -A_{1}^{-1}B_1 \right) W_{k}^{x}, \nonumber \\
%	Y_{k-1}&=A_{2}^{-1}X_k+\left( -A_{2}^{-1}A_1 \right) X_{k-1}+\left( -A_{2}^{-1}B_1 \right) W_{k}^{x}. \nonumber
%\end{align}
According to Lemma \ref{Lm.basicInterval}, the tightest intervals for \eqref{Eq.exactYEst0} and \eqref{Eq.exactYEst1} are \eqref{Eq.TightInterval1} and \eqref{Eq.TightInterval2}. Then, the intersection of different intervals, i.e., \eqref{Eq.exactYEst2} and \eqref{Eq.exactYEst3}, can be computed with \eqref{Eq.TightInterval3} and \eqref{Eq.TightInterval4}. Finally, the one-step forward reachable set  \eqref{Eq.exactYEst4} can be approximated with the tightest interval \eqref{Eq.TightInterval5} based on Lemma \ref{Lm.basicInterval}, and the calibrated uncertainty set $\mathcal{X}_{k|k}$ and the tightest prior inference set $\mathcal{Y}_{k|k-1}$ can be approximated similarly. 
\section{Proof of Lemma \ref{Lm.radiusY}}\label{App.Lm.radiusY}
According to Lemma \ref{Lm.basicInterval}, the radius of $\mathcal{M}^y_{k-1|k}$ can be computed as
\begin{align}\label{Eq.radiusYMK_1MultiDim}
	p_{k-1|k}^{m,y}\!=\! \left|\! A_{2}^{-1}\! \right|p_{k|k}^{m,x}\!+\!\left|\! A_{2}^{\!-\!1}\!A_1 \!\right|p_{k\!-\!1|k-\!1}^{x}\!+\!\left|\! A_{2}^{\!-\!1}\!B_1 \!\right|p_{k}^{w,x}\!,\!\!
\end{align}
where $p^{m,x}_{k|k}$, $p_{k-1|k-1}^{x}$ and $p^{w,x}_k$ are radii of $\mathcal{M}_{k|k}^x$, $\mathcal{X}_{k-1|k-1}$ and $\mathcal{W}_{k}^x$, respectively.
Since $\mathcal{Y}_{k-1|k}$ is the intersection result from $\mathcal{M}^y_{k-1|k}$ and $\mathcal{Y}_{k-1|k-1}$, the radius of $\mathcal{Y}_{k-1|k}$ is smaller than the radius of $\mathcal{M}^y_{k-1|k}$, i.e., $p_{k-1|k}^{y}\leqslant p_{k-1|k}^{m,y}$. 
Also, the radius of $\mathcal{Y}_{k|k}$ can be computed as
\begin{align}\label{Eq.radiusYKMultiDim}
	p_{k|k}^{y}=\left| A_3 \right|p_{k-1|k}^{x}+\left| A_4 \right|p_{k-1|k}^{y}+\left| B_2 \right|p_{k}^{w,y}. 
\end{align}
By substituting \eqref{Eq.radiusYMK_1MultiDim} and $p_{k-1|k}^{y}\leqslant p_{k-1|k}^{m,y}$ into \eqref{Eq.radiusYKMultiDim}, we have
\begin{align}
	p_{k|k}^{y}\leqslant \left| A_3 \right|p_{k-1|k}^{x}+\left| B_2 \right|p_{k}^{w,y}+\left| A_4 \right|\left| A_{2}^{-1}B_1 \right|p_{k}^{w,x}\nonumber\\+\left| A_4 \right|\left( \left| A_{2}^{-1} \right|p_{k|k}^{m,x}+\left| A_{2}^{-1}A_1 \right|p_{k-1|k-1}^{x} \right) . \nonumber
\end{align}
Since $\overline{p}^x\geqslant p^{m,x}_{j|j}$ for any $j\geqslant0$ and $\mathcal{X}_{k-1|k}$ is a subset of $\mathcal{M}_{k-1|k-1}^x$, we have $p_{k-1|k}^{x}\leqslant p^{m,x}_{k-1|k-1}\leqslant\overline{p}^x$ for any $k\geqslant1$, thus we have \eqref{Eq.radiusYInequality}.
\section{Proof of Theorem \ref{Th.PriLeakBounds}}\label{App.Th.PriLeakBounds}
The difference set $\Delta \mathcal{Y}_{k|k}$ is computed as, $$\Delta \mathcal{Y} _{k|k}=\mathcal{Y} _{k|k-1}\setminus \mathcal{Y} _{k|k}=\Phi\left(A_3\right)\Delta \mathcal{X} _{k-1|k}\oplus  \Phi\left(A_4\right)\Delta \mathcal{Y} _{k-1|k},$$
where 
\begin{align}
	\Delta \mathcal{X} _{k-1|k}&=\mathcal{X} _{k-1|k-1}\setminus \mathcal{X} _{k-1|k}\nonumber\\&=\left[ \begin{array}{c}
		\min \left\{ \underline{X}_{k-1|k-1}-\underline{M}_{k-1|k}^{x},0 \right\}\\
		\max \left\{ \overline{X}_{k-1|k-1}-\overline{M}_{k-1|k}^{x},0 \right\}\\
	\end{array} \right] , \nonumber \\ 
	\Delta \mathcal{Y} _{k-1|k}&=\mathcal{Y} _{k-1|k-1}\setminus \mathcal{Y} _{k-1|k}\nonumber\\&=\left[ \begin{array}{c}
		\min \left\{ \underline{Y}_{k-1|k-1}-\underline{M}_{k-1|k}^{y},0 \right\}\\
		\max \left\{ \overline{Y}_{k-1|k-1}-\overline{M}_{k-1|k}^{y},0 \right\}\\
	\end{array} \right] .  \nonumber
\end{align}
Therefore, the volume of the difference set is \eqref{Eq.PrivacyLeaksMultiDimen}.

With Lemma \ref{Lm.basicInterval}, we have
\begin{align}
	p_{k|k}^{\Delta y}=\left| A_3 \right|p_{k-1|k}^{\Delta x}+\left| A_4 \right|p_{k-1|k}^{\Delta y}, \nonumber
\end{align}
where the radius $p_{k-1|k}^{\Delta x}$ and $p_{k-1|k}^{\Delta y}$ can be computed via
\small
\begin{align}\label{Eq.positiveDelatK_1}
	&2p_{k-1|k}^{\Delta z} \nonumber
	\\=&\max \left\{ \overline{Z}_{k-1|k-1}-\overline{M}_{k-1|k}^{z},0 \right\} -\min \left\{ \underline{Z}_{k-1|k-1}-\underline{M}_{k-1|k}^{z},0 \right\} \nonumber
	\\
	=&\max \left\{0, \overline{Z}_{k-1|k-1}-\overline{M}_{k-1|k}^{z}+\underline{M}_{k-1|k}^{z}-\underline{Z}_{k-1|k-1}, \right. \nonumber \\
	&\!\!\!\!\overline{Z}_{k-\!1|k-\!1}\!-\!\overline{M}_{k-\!1|k}^{z},\underline{M}_{k-1|k}^{z}\!-\!\underline{Z}_{k-\!1|k-\!1} \left.\right\}, \text{for}\, Z=X,Y, \!\!
\end{align}
\normalsize
which satisfies $p_{k-1|k}^{\Delta z} \geqslant 0$. As a result, we have
\begin{align}
	\overline{\mathrm{Vol}}\left( \Delta \mathcal{Y} _{k|k} \right) 
	&=\left\| \left| A_3 \right|p_{k-1|k}^{\Delta x}+\left| A_4 \right|p_{k-1|k}^{\Delta y} \right\|_1 \nonumber
	\\
	&\overset{(a)}{\leqslant} \left\| A_3 \right\|_1\left\| p_{k-1|k}^{\Delta x} \right\|_1+\left\| A_4 \right\|_1\left\| p_{k-1|k}^{\Delta y} \right\|_1 \nonumber
	\\
	&=\left\| A_3 \right\|_1\overline{\mathrm{Vol}}\left( \Delta \mathcal{X} _{k-1|k} \right) +\left\| A_4 \right\|_1\overline{\mathrm{Vol}}\left( \Delta \mathcal{Y} _{k-1|k} \right) , \nonumber
\end{align}
where $(a)$ is due to $p_{k-1|k}^{\Delta x} \geqslant 0$ and $p_{k-1|k}^{\Delta y} \geqslant 0$.

Besides, given an interval $\mathcal{X}$, we can express it with its center point and radius, i.e., $\mathcal{X} = \left[\begin{array}{c}
	\frac{c-p}{2}\\ \frac{c+p}{2}
\end{array}\right]$. Therefore, we have 
\begin{align}\label{Eq.deltaVolYLower}
	\overline{\mathrm{Vol}}\left( \Delta \mathcal{Y} _{k|k} \right) = &\left\|\overline{Y} _{k|k}-\overline{Y} _{k|k-1}\right\|_1 + \left\|\underline{Y} _{k|k}-\underline{Y} _{k|k-1}\right\|_1  \nonumber
	\\
	\geqslant & \left\| \overline{Y} _{k|k}+\underline{Y} _{k|k}-\left(\overline{Y} _{k|k-1}+\underline{Y} _{k|k-1}\right) \right\|_1 \nonumber
	\\
	\geqslant & 2\left\| c_{k|k}^{y}-c_{k|k-1}^{y} \right\|_1. \nonumber
\end{align}
\section{Proof of Theorem \ref{Th.linearOpt}}\label{App.Th.linearOpt}
To maximize the privacy level, it is equivalent to minimize the amount of uncertainty reduction since we have $\mathrm{Vol}\left(\Delta \mathcal{Y}_{k|k}\right) = \mathrm{Vol}\left( \mathcal{Y}_{k|k-1}\right)-\mathrm{Vol}\left( \mathcal{Y}_{k|k}\right)$, where the prior uncertainty set $\mathcal{Y}_{k|k-1}$ is fixed at time step $k$. 

Besides, the amount of uncertainty reduction
$\mathrm{Vol}\left(\Delta\mathcal{Y}_{k|k}\right) = \left\|p_{k|k}^{\Delta y}\right\|_1= \left\| \left| A_3 \right|p_{k-1|k}^{\Delta x}+\left| A_4 \right|p_{k-1|k}^{\Delta y} \right\|_1$, where the elements of $p_{k-1|k}^{\Delta x}$ and $p_{k-1|k}^{\Delta y}$ are non-negative vectors as shown in \eqref{Eq.positiveDelatK_1}. Therefore, we can replace the objective function with the slack variable $\epsilon^y$ and add $\mathrm{Vol}\left(\Delta\mathcal{Y}_{k|k}\right)\leqslant\epsilon ^y$ as a new constraint, and then minimize $\epsilon^y$.

Since $\overline{\mathrm{Vol}}\left(\Delta\mathcal{Y}_{k|k}\right)$ is determined by $p_{k-1|k}^{\Delta x}$ and $p_{k-1|k}^{\Delta y}$, we can replace constraints \eqref{Eq.TightInterval3} and \eqref{Eq.TightInterval4} with the constraints of difference sets \eqref{Eq.positiveDelatK_1}. Also, the objective function increases with any elements of $p_{k-1|k}^{\Delta x}$ and $p_{k-1|k}^{\Delta y}$ since the elements of $p_{k-1|k}^{\Delta x}$, $p_{k-1|k}^{\Delta y}$, $\left|A_3\right|$ and $\left|A_4\right|$ are non-negative. As a result, we can replace the constraint of $p_{k-1|k}^{\Delta x}$ and $p_{k-1|k}^{\Delta y}$, i.e., \eqref{Eq.positiveDelatK_1}, with inequalities \eqref{Eq.optDeltaConstraint1}, and let $p_{k-1|k}^{\Delta x} $ and $ p_{k-1|k}^{\Delta y}$ be decision variables. 

Besides, the constraints $\mathcal{S}_{k|k}^x \subseteq \mathcal{M}^x _{k|k}$, $\mathcal{M}^x _{k|k}\subseteq \mathcal{X} _{k|k-1}$ and $\overline{M}^x_{k|k} \geqslant \underline{M}^x_{k|k}$ are equivalent to the inequality constraint,  $\underline{X}_{k|k-1}\leqslant \underline{M}^x_{k|k}\leqslant \underline{S}^x_{k|k} \leqslant \overline{S}^x_{k|k}\leqslant \overline{M}^x_{k|k}\leqslant \overline{X}_{k|k-1}$, and the utility constraint $\overline{\mathrm{Vol}}\left( \mathcal{M}^x _{k|k} \right) \leqslant \epsilon^x$ can be replaced with $\left\| \overline{M}^x_{k|k}-\underline{M}^x_{k|k} \right\|_1\leqslant \epsilon ^x$. 

Finally, the objective and the constraints are linear functions of the decision variables, thus, the optimal privacy filter can be obtained by solving the linear programming $\mathbf{P_2}$.
%\begin{ack}                               % Place acknowledgements
%Partially supported by the Roman Senate.  % here.
%\end{ack}

\bibliographystyle{unsrt}        % Include this if you use bibtex 
\bibliography{reference}           % and a bib file to produce the 
                                 % bibliography (preferred). The
                                 % correct style is generated by
                                 % Elsevier at the time of printing.

%\begin{thebibliography}{99}     % Otherwise use the  
                                 % thebibliography environment.
                                 % Insert the full references here.
                                 % See a recent issue of Automatica 
                                 % for the style.
%  \bibitem[Heritage, 1992]{Heritage:92}
%     (1992) {\it The American Heritage. 
%     Dictionary of the American Language.}
%     Houghton Mifflin Company.
%  \bibitem[Able, 1956]{Abl:56}
%     B.~C.~Able (1956). Nucleic acid content of macroscope. 
%     {\it Nature 2}, 7--9. 
%  \bibitem[Able {\em et al.}, 1954]{AbTaRu:54}   
%     B.~C. Able, R.~A. Tagg, and M.~Rush (1954).
%     Enzyme-catalyzed cellular transanimations.
%     In A.~F.~Round, editor, 
%     {\it Advances in Enzymology Vol. 2} (125--247). 
%     New York, Academic Press.
%  \bibitem[R.~Keohane, 1958]{Keo:58}
%     R.~Keohane (1958).
%     {\it Power and Interdependence: 
%     World Politics in Transition.}
%     Boston, Little, Brown \& Co.
%  \bibitem[Powers, 1985]{Pow:85}
%     T.~Powers (1985).
%     Is there a way out?
%     {\it Harpers, June 1985}, 35--47.

%\end{thebibliography}

                                        % in the appendices.
\end{document}